\documentclass[journal,compsoc,10pt]{IEEEtran}

\usepackage{amsmath,amssymb,amsthm}
\usepackage{graphicx}
\usepackage{url}
\usepackage{enumitem}
\usepackage{bm}
\usepackage[textsize=small]{todonotes}
\usepackage{subfigure}
\usepackage{ifthen}
\usepackage[all]{nowidow}
\usepackage{thmtools}
\usepackage{thm-restate}

\newcommand{\newparentheses}[3]{%
  \expandafter\newcommand\csname #1\endcsname[1]{#2##1#3}%
  \expandafter\newcommand\csname #1L\endcsname[1]{\bigl#2##1\bigr#3}%
  \expandafter\newcommand\csname #1XL\endcsname[1]{\Bigl#2##1\Bigr#3}%
  \expandafter\newcommand\csname #1V\endcsname[1]{\left#2##1\right#3}}

\newparentheses{parens}{(}{)}
\newparentheses{braces}{\{}{\}}
\newparentheses{brackets}{[}{]}
\newparentheses{floor}{\lfloor}{\rfloor}
\newparentheses{ceil}{\lceil}{\rceil}
\newparentheses{abs}{|}{|}
\newparentheses{set}{\{}{\}}
\newparentheses{size}{|}{|}

\makeatletter
\newcommand{\onenewattribute}[3]{%
  \@ifundefined{#1}{\let\@@def\newcommand}{\let\@@def\renewcommand}%
  \expandafter\@@def\csname #1\endcsname[2][]{%
    \ifthenelse{\equal{##1}{}}%
    {#2\csname #3\endcsname{##2}}%
    {#2_{##1}\csname #3\endcsname{##2}}}}
\newcommand{\newattribute}[2]{%
  \onenewattribute{#1}{#2}{parens}%
  \onenewattribute{#1L}{#2}{parensL}%
  \onenewattribute{#1XL}{#2}{parensXL}%
  \onenewattribute{#1V}{#2}{parensV}}
\makeatother

\newattribute{OhOf}{\mathrm{O}}
\newattribute{ThetaOf}{\Theta}
\newattribute{OmegaOf}{\Omega}
\newattribute{ohOf}{\mathrm{o}}
\newattribute{omegaOf}{\omega}

\newattribute{alg}{\textsf{mAF}}

\newattribute{algstar}{\textsf{mAF*}}
\newattribute{replug}{\textsf{replug}}
\newattribute{replugdecorate}{\textsf{replug-decorate}}
\newattribute{replugdecoratetwo}{\textsf{replug-decorate*}}

\newattribute{dspr}{d_{\mathrm{SPR}}}
\newattribute{atbr}{a_{\mathrm{TBR}}}
\newattribute{dtbr}{d_{\mathrm{TBR}}}
\newattribute{dreplug}{d_{\mathrm{R}}}
\newattribute{ecut}{e}
\newattribute{mafsize}{m}
\newattribute{maf}{\textrm{MAF}}
\newcommand{\one}{\mathrm{ONE}}

\newcommand{\edge}[1]{e_{#1}}
\newcommand{\eedge}[2]{#1\rhd\!#2}

\newcommand{\weight}{\omega}

\newcommand{\reach}[1]{\sim_{#1}}
\newcommand{\noreach}[1]{\not\sim_{#1}}

\newcommand{\quartet}[2]{#1|#2}

\newcommand{\subtree}[2][]{%
  \ifthenelse{\equal{#1}{}}%
  {T(#2)}%
  {#1(#2)}}

\newcommand{\induced}[2][]{%
  \ifthenelse{\equal{#1}{}}%
  {T|#2}%
  {#1|#2}}

\newcommand{\cnfletwo}{\text{CNF}_{+} (\le 2)}

\newcommand{\notarxiv}[1]{}
\newcommand{\eat}[1]{}

\declaretheorem[numberwithin=section, name=Theorem]{theorem}
\declaretheorem[sibling=theorem, name=Lemma]{lemma}
\declaretheorem[sibling=theorem, name=Corollary]{corollary}
\declaretheorem[sibling=theorem, name=Observation]{obs}
\declaretheorem[sibling=theorem, name=Conjecture]{conjecture}

\declaretheorem[sibling=theorem, name=Property]{property}

\newcommand{\beginsupplement}{%
        \setcounter{table}{0}
        \renewcommand{\thetable}{S\arabic{table}}%
        \setcounter{figure}{0}
        \renewcommand{\thefigure}{S\arabic{figure}}%
     }

\begin{document}

\title{Calculating the Unrooted Subtree Prune-and-Regraft Distance}
\author{Chris Whidden and Frederick A. Matsen IV
\IEEEcompsocitemizethanks{%
	\IEEEcompsocthanksitem C. Whidden and F. Matsen are with the Program in Computational Biology, Fred Hutchinson Cancer Research Center, Seattle, WA, 98109. E-mail: \{cwhidden,matsen\}@fredhutch.org}
}




\IEEEtitleabstractindextext{%
\begin{abstract}%
The subtree prune-and-regraft (SPR) distance metric is a fundamental way of comparing evolutionary trees.
It has wide-ranging applications, such as to study lateral genetic transfer, viral recombination, and Markov chain Monte Carlo phylogenetic inference.
Although the rooted version of SPR distance can be computed relatively efficiently between rooted trees using fixed-parameter-tractable maximum agreement forest (MAF) algorithms, no MAF formulation is known for the unrooted case.
Correspondingly, previous algorithms are unable to compute unrooted SPR distances larger than 7.

In this paper, we substantially advance understanding of and computational algorithms for the unrooted SPR distance.
First we identify four properties of optimal SPR paths, each of which suggests that no MAF formulation exists in the unrooted case.
Then we introduce the replug distance, a new lower bound on the unrooted SPR distance that is amenable to MAF methods, and give an efficient fixed-parameter algorithm for calculating it.
Finally, we develop a ``progressive A*'' search algorithm using multiple heuristics, including the TBR and replug distances, to exactly compute the unrooted SPR distance.
Our algorithm is nearly two orders of magnitude faster than previous methods on small trees, and allows computation of unrooted SPR distances as large as 14 on trees with 50 leaves.
\end{abstract}

\begin{IEEEkeywords}
fixed-parameter tractability, phylogenetics, subtree prune-and-regraft distance, lateral gene transfer, agreement forest.
\end{IEEEkeywords}}

\maketitle

\emergencystretch=1em

\IEEEraisesectionheading{\section{Introduction}
\label{sec:intro}}
\IEEEPARstart{M}{olecular} phylogenetic methods reconstruct evolutionary trees (a.k.a\ phylogenies) from DNA or RNA data and are of fundamental importance to modern biology~\cite{hillis96}.
Phylogenetic inference has numerous applications including investigating organismal relationships (the "tree of life"~\cite{koonin2015turbulent}), reconstructing virus evolution away from innate and adaptive immune defenses~\cite{castro2012evolution}, analyzing the immune system response to HIV~\cite{haynes2012b}, designing genetically-informed conservation measures~\cite{helmus2007phylogenetic}, and investigating the human microbiome~\cite{matsen2015phylogenetics}.
Although the molecular evolution assumptions may differ for these different settings, the core algorithmic challenges remain the same: reconstruct a tree graph representing the evolutionary history of a collection of evolving units, which are abstracted as a collection of \emph{taxa}, where each \emph{taxon} is associated with a DNA, RNA, or amino acid sequence.

Phylogenetic study often requires an efficient means of comparing phylogenies in a meaningful way.
For example, different inference methods may construct different phylogenies and it is necessary to determine to what extent they differ and, perhaps more importantly, which specific features differ between the trees.
In addition, the evolutionary history of individual genes does not necessarily follow the overall history of a species due to \emph{reticulate} evolutionary processes: lateral genetic transfer, recombination, hybridization, and incomplete lineage sorting~\cite{galtier2008dealing}.
Such processes impact phylogenies by moving a subtree from one location to another, as described below.
Thus comparing inferred histories of genes to each other, a reference tree, or a proposed species tree may be used to identify reticulate events~\cite{beiko2005highways,whidden2014supertrees}.
Moreover, distance measures between phylogenies provide optimization criteria that can be used to infer summary measures such as supertrees~\cite{pisani2007supertrees,Steel2008-pn,bansal2010robinson,whidden2014supertrees}.

Numerous distance measures have been proposed for comparing phylogenies.
The Robinson-Foulds distance~\cite{robinson81} is perhaps the most well known and can be calculated in linear time~\cite{day85}.
However, the Robinson-Foulds distance has no meaningful biological interpretation or relationship to reticulate evolution.
Typically, distance metrics are either easy to compute but share this lack of biological relation, such as the quartet distance~\cite{brodal2004computing} and geodesic distance~\cite{owen2011fast}, or are difficult to compute, such as the hybridization number~\cite{baroni05} and maximum parsimony distance~\cite{Bruen2008-wx,kelk2014complexity,moulton2015parsimony}.

The subtree prune-and-regraft (SPR) distance is widely used due to its biological interpretability despite being difficult to compute~\cite{baroni05,beiko2006phylogenetic}.
SPR distance is the minimum number of subtree moves required to transform one tree into the other (Figure~\ref{fig:spr}).
It provides a lower bound on the number of reticulation events required to reconcile two phylogenies.
As such, it has been used to model reticulate evolution~\cite{maddison97,nakhleh05}.
In addition, the SPR distance is a natural measure of distance when analyzing phylogenetic inference methods which typically apply SPR operations to find maximum likelihood trees~\cite{Price2010-fi,Stamatakis2006-yz} or estimate Bayesian posterior distributions with SPR-based Metropolis-Hastings random walks~\cite{Ronquist2012-hi,bouckaert2014beast}.
Similar trees can be easily identified using the SPR distance, as random pairs of $n$-leaf trees differ by by an expected $n - \Theta(n^{2/3})$ SPR moves~\cite{atkins2015extremal}.
This difference approaches the maximum SPR distance of $n - 3 - \floorV{(\sqrt{n-2} -1)/2}$ asymptotically \cite{Ding2011-bj}.
The topology-based SPR distance is especially appropriate in this context as topology changes have been identified as the main limiting factor of such methods~\cite{lakner2008efficiency,hohna2012guided,whidden2015quantifying}.
Moreover, the SPR distance has close connections to network models of evolution~\cite{baroni05,bordewich07,nakhleh05}.

Although it has these advantages, the SPR distance between both rooted and unrooted trees is NP-hard to compute~\cite{bordewich2005computational,hickey2008spr}, limiting its utility.
We recall that \emph{rooted} trees represent the usual view of evolution, such that the taxa under consideration evolve from a common ancestor in a known direction.
\emph{Unrooted} trees drop this implied directionality, and are typically drawn as a graph theoretic tree such that every non-leaf node has degree three (Figure~\ref{fig:x-tree}).
As described below, most phylogenetic algorithms reconstruct unrooted trees.

Despite the NP-hardness of computing the SPR distance between rooted phylogenies, recent algorithms can rapidly compare rooted trees with hundreds of leaves and SPR distances of 50 or more in fractions of a second~\cite{whidden2010fast,whidden2014supertrees}.
This has enabled use of the SPR distance for inferring phylogenetic supertrees and lateral genetic transfer~\cite{whidden2014supertrees}, comparing influenza phylogenies to assess reassortment~\cite{dudas2014reassortment}, and investigating mixing of Bayesian phylogenetic posteriors~\cite{whidden2015quantifying,whidden2015ricci}.
However, most phylogenetic inference packages today use reversible mutation models to infer unrooted trees, motivating SPR calculation for unrooted trees.

SPR distances can be computed efficiently in practice for rooted trees for two key reasons.
First, they can be computed using a maximum agreement forest (MAF) of the trees~\cite{hein96,allen01}.
An MAF is a forest (i.e.\ collection of trees) that can be obtained from both trees by removing a minimum set of edges.
Each removed edge corresponds to one SPR operation, and a set of SPR operations transforming one tree into the other can be easily recovered given an MAF.
Due to this MAF framework, the development of efficient fixed-parameter and approximation algorithms for SPR distances between rooted trees has become an area of active research~\cite{beiko2006phylogenetic,wu2009practical,bonet2009efficiently,whidden2013fixed,shi2014improved,chen2013faster} (see ~\cite{whidden2013fixed} for a more complete history), including recent extensions to nonbinary trees~\cite{whidden2015multifurcating,chen2015parameterized}, as well as generalized MAFs of multiple trees~\cite{shi2014approximation}.
The utility of MAFs motivates defining a variant definition in the unrooted case with analogous properties for unrooted SPR, however, no alternative MAF formulation has yet been developed.
A straightforward extension of MAFs to unrooted trees is equivalent to a different metric, the TBR distance~\cite{allen01}.
Although TBR rearrangements are used in some phylogenetic inference methods, SPR rearrangements are much more common~\cite{lakner2008efficiency} and the TBR distance does not have the other benefits of the SPR distance.

The second class of optimizations used by efficient rooted SPR algorithms are preprocessing reduction rules including the subtree reduction rule~\cite{allen01}, chain reduction rule~\cite{allen01}, and cluster reduction rule~\cite{linz2011cluster}.
The subtree reduction rule also applies to unrooted trees~\cite{allen01} and we recently showed that the chain reduction rule is applicable to the unrooted case~\cite{whidden2016chain}, thereby obtaining a linear-size problem kernel for unrooted SPR.
However, minimum-length uSPR paths have been shown to break common clusters~\cite{hickey2008spr}, so the cluster reduction, which partitions the trees into smaller independently solvable subproblems, is not applicable.

For all of these reasons, the best previous algorithm for computing the SPR distance between unrooted trees, due to Hickey et al.~\cite{hickey2008spr}, cannot compute distances larger than 7 or reliably compare trees with more than 30 leaves.
In this paper, we substantially advance understanding of and computational algorithms for the unrooted SPR distance.
Building on previous work by Hickey et al.\ \cite{hickey2008spr}, Bonet and St. John \cite{bonet2010complexity}, and our recently introduced ``socket agreement forest (SAF)'' framework~\cite{whidden2016chain}, we make the following contributions:

\begin{enumerate}
	\item we identify new properties of minimum-length SPR paths showing that an MAF-like formulation is unlikely to~exist,
	\item we develop a practical algorithm for enumerating maximal unrooted AFs,
	\item we define a new \emph{replug} distance, which does admit a MAF-like formulation and gives a lower bound on the uSPR distance; we develop an exact fixed-parameter bounded search tree algorithm for its calculation, and
	\item we propose and implement a new incremental heuristic search algorithm called \emph{progressive A*} that leverages multiple increasingly expensive to compute but more accurate lower bound estimators to compute the uSPR distance in practice for trees with up to 50 leaves and distances as large as 14.
\end{enumerate}

Proofs of our lemmas and theorems are in the appendix.

\section{Preliminaries}
\label{sec:prelim}

\begin{figure}[t]
	\subfigure[\label{fig:x-tree}]{\includegraphics[scale=1]{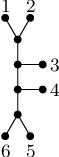}}
	\hspace*{\stretch{1}}
	\subfigure[\label{fig:subtree}]{\includegraphics[scale=1]{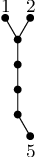}}
	\hspace*{\stretch{1}}
	\subfigure[\label{fig:induced}]{\includegraphics[scale=1]{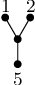}}
	\hspace*{\stretch{2}}
	\subfigure[\label{fig:spr}]{\includegraphics[scale=1]{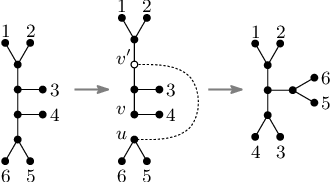}}

	\caption{(a) An unrooted $X$-tree $T$.
		(b) $T(V)$, where $V = \set{1,2,5}$.
		(c) $T|V$.
		(d) An SPR operation transforms $T$ into a new tree by \emph{pruning a subtree} and \emph{regrafting} it in an other location.
	}
	\label{fig:trees}
\end{figure}

Nodes (i.e. vertices) of a tree graph with one neighbor are called \emph{leaves} and nodes with three neighbors are called \emph{internal nodes}.
An (unrooted binary phylogenetic) \emph{$X$-tree} is a tree $T$ whose nodes each have one or three neighbors, and such that the leaves of $T$ are bijectively labeled with the members of a label set $X$.
$T(V)$ is the unique subtree of $T$ with the fewest nodes that connects all nodes in $V$.
\emph{Suppressing} a node $v$ of degree 1 or 2 deletes $v$ and its incident edges; if $v$ is of degree 2 with neighbors $u$ and $w$, $u$ and $w$ are reconnected using a new edge $(u,w)$.
The $V$-tree \emph{induced} by $T$ is the unique smallest tree $T|V$ that can be obtained from $T(V)$ by suppressing unlabeled nodes with fewer than three neighbors.
See Figure~\ref{fig:trees}.

A \emph{rooted $X$-tree} is defined similarly to an unrooted $X$-tree, with the exception that one of the internal nodes is called the \emph{root} and is adjacent to a leaf labeled $\rho$.
Note that this differs from the standard definition of a rooted tree in which the root is a degree two internal node.
This $\rho$ node represents the position of the original root in a forest of the trees, as we describe below, and can simply be attached to such a degree two internal node.
The parent of a node in an rooted tree is its closest neighbor to the root; the other two neighbors of an internal node are referred to as children.

An \emph{unrooted $X$-forest} $F$ is a collection of (not necessarily binary) trees $T_1, T_2, \ldots,\allowbreak T_k$ with respective label sets $X_1, X_2, \ldots, X_k$ such that $X_i$ and $X_j$ are disjoint, for all $1 \le i \ne j \le k$, and $X = X_1 \cup X_2 \cup \ldots \cup X_k$.
We say $F$ \emph{yields} the forest with components $T_1|X_1$, $T_2|X_2$, \ldots, $T_k|X_k$, in other words, this forest is the smallest forest that can be obtained from $F$ by suppressing unlabeled nodes with less than three neighbors.
In the rooted case, each component $T_i$ is rooted at the node that was closest to $\rho$.
Note that only the root of $T_1$ is adjacent to leaf $\rho$.
If $T_1, T_2, \ldots, T_k$ are all binary then the remaining roots have degree 2.
For an edge set $E$, $F-E$ denotes the forest obtained by deleting the edges in $E$ from $F$ and $F \div E$ the yielded forest.
We say that $F \div E$ is \emph{a forest of $F$}.
For nodes $a$ and $b$ of $F$, we will say that $a$ can be reached from $b$, or $a \reach{F} b$, when there is a path of edges between $a$ and $b$ in $F$.
The opposite will be denoted $a \noreach{F} b$.

A \emph{subtree-prune-regraft} (uSPR) operation on an unrooted $X$-tree $T$ cuts an edge $e = (u,v)$.
This divides $T$ into subtrees $T_u$ and $T_v$, containing $u$ and $v$, respectively.
Then it introduces a new node $v'$ into $T_v$ by subdividing an edge of $T_v$, and adds an edge $(u,v')$.
Finally, $v$ is suppressed (Figure~\ref{fig:spr}).
We distinguish between SPR operations on rooted trees and unrooted trees (uSPR operations).
SPR operations on rooted trees have the additional requirement that $u \ne \rho$ and that $v$ is $u$'s parent rather than an arbitrary neighbor of $u$.
Note that if the node $v'$ introduced in the rooted tree is adjacent to $\rho$ then $v'$ becomes the root.

A \emph{tree-bisection-reconnection} (TBR) operation on an unrooted tree is defined similarly to a uSPR operation, except that a new node $u'$ is also introduced into $T_u$ bisecting any edge, the added edge is $(u',v')$ rather than $(u,v')$, and both $u$ and $v$ are suppressed.
Note that uSPR operations are a subset of TBR operations, as a TBR operation may reintroduce the same endpoint on one side of the edge.

We often consider a sequence of operations applied to a tree $T_1$ that result in a tree $T_2$.
These operations can be thought of as ``moving'' between trees and are also referred to as \emph{moves} (e.g. an SPR move).
A sequence of moves $M = m_1, m_2, \ldots, m_d$ applied to $T_1$ result in the sequence of trees $T_1 = t_0, t_1, t_2, \ldots, t_d = T_2$.
We call such sequences of trees a \emph{path} (e.g. an SPR path).

When considering how the tree changes throughout such sequences, it is often helpful to consider how nodes and edges of the tree change.
Formally, we construct a mapping $\varphi_{i,j}$ that maps nodes and edges of $t_i$ to $t_j$.
Each mapping $\varphi_{i,i+1}$ between adjacent trees is constructed according to the corresponding move $m_{i+1}$: nodes and edges of $t_i$ that are not modified by $m_{i+1}$ are mapped to the corresponding nodes and edges of $t_{i+1}$.
The deleted edge $(u,v)$ of $t_i$ is mapped to the newly introduced edge of $t_{i+1}$ (e.g. $(u,v')$ for an SPR move).
Deleted nodes are mapped to $\emptyset$.
Forward mappings $\varphi_{i,j}$, $i < j$, are constructed transitively.
Reverse mappings $\varphi_{j,i}$, $i < j$, are constructed analogously by considering the application of moves that construct the reverse sequence $t_d, t_{d-1}, \ldots, t_0$.

We will use these mappings implicitly to talk about how a tree changes throughout a sequence of moves.
With these mappings we can consider SPR and TBR tree moves as changing the endpoints of edges rather than deleting one edge and introducing another.
We say that an edge is \emph{broken} if one of its endpoints is moved by a rearrangement operation.
The relation between SPR and TBR moves can now be summarized by considering that a uSPR operation changes one endpoint of an edge, a rooted tree SPR operation changes the root-most endpoint of an edge, and a TBR operation changes both endpoints.

\begin{figure}[t]
	\hspace*{\stretch{1}}
	\subfigure[\label{fig:three-spr}]{\includegraphics[scale=1]{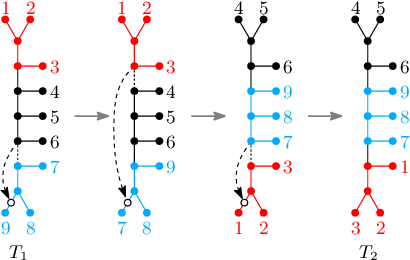}}
	\hspace*{\stretch{2}}
	\subfigure[\label{fig:maf}]{\includegraphics[scale=1]{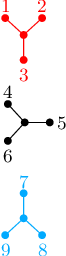}}
	\hspace*{\stretch{1}}

	\caption{(a) Three unrooted SPR operations transform a tree $T_1$ into another tree $T_2$. The dashed line with an arrow is the edge that will be added in the next step; simultaneously the dotted edge touching the tail end of the arrow will be removed.
		(b) A MAF of $T_1$ and $T_2$.
	}
	\label{fig:spr-distance}
\end{figure}

SPR operations give rise to a distance measure $\dspr{\cdot,\cdot}$ between $X$-trees, defined as the minimum number of SPR operations required to transform one tree into the other.
The trees in Figure~\ref{fig:three-spr}, for example, have SPR distance $\dspr{T_1,T_2} = 3$.
The TBR distance $\dtbr{\cdot,\cdot}$ is defined analogously with respect to TBR operations.
A minimum-length path of SPR or TBR moves between two trees is an optimal SPR path or optimal TBR path, respectively.

The second and third SPR operations applied in Figure~\ref{fig:three-spr} move both ends of a single edge and can be replaced by a single TBR operation.
Thus, in this case $\dtbr{T_1, T_2} = 2$.
Note that the fact that uSPR operations are a subset of TBR operations implies that the TBR distance is a lower bound on the uSPR distance~\cite{allen01}:

\begin{lemma}
	For two unrooted trees $T_1$ and $T_2$, $\dtbr{T_1, T_2} \le \dspr{T_1, T_2}$.
\end{lemma}

Given trees $T_1$ and $T_2$ and forests $F_1$ of $T_1$ and $F_2$ of $T_2$, a forest $F$ is an \emph{agreement forest} (AF) of $F_1$ and $F_2$ if it is a forest of both forests.
$F$ is a \emph{maximum agreement forest} (MAF) if it has the smallest possible number of components.
We denote this number of components by $m(F_1, F_2)$.
For two unrooted trees $T_1$ and $T_2$, Allen and Steel~\cite{allen01} showed that $\dtbr{T_1,T_2} = m(T_1, T_2) - 1$.
Figure~\ref{fig:maf} shows an MAF of the trees in Figure~\ref{fig:three-spr}.

For two \emph{rooted} trees $T_1$ and $T_2$, Bordewich and Semple~\cite{bordewich2005computational} showed $\dspr{T_1,T_2} = m(T_1, T_2) - 1$, by introducing the root node augmentation $\rho$ described above.
Moreover, nodes are only suppressed in rooted trees if they have fewer than two children, such that unlabeled \emph{component roots} (the nodes that were connected to the $\rho$ component by an edge before cutting) are not suppressed.

In contrast to an MAF, a \emph{maximal AF (mAF)} $F^*$ of two trees $T_1$ and $T_2$ is an AF of $T_1$ and $T_2$ that is not a forest of any other AF of $T_1$ and $T_2$.
In other words, no edges can be added to $F^*$ to obtain an AF with fewer components.
Every MAF is a mAF, but not necessarily vice versa~\cite{shi2013parameterized}.

\section{Four strikes against an agreement forest for the uSPR distance}
\label{sec:properties}

There are relatively efficient fixed-parameter algorithms for each distance metric with a maximum (acyclic) agreement forest formulation---the rooted SPR distance~\cite{hein96}, unrooted TBR distance~\cite{allen01}, and rooted hybridization number~\cite{baroni05}.
We call these MAF-like problems.
For MAF-like problems, the distance can be found from easily determined properties of the maximum agreement forest: typically, its number of components.
It is reasonable to ask why there has been no formulation of the uSPR distance as an MAF-like problem.
In this section, we identify four properties of the uSPR distance that are contrary to natural MAF assumptions.
Specifically, these properties show the uSPR distance isn't easily calculated from properties of a traditionally defined agreement forest.

Let $S$ be an optimal sequence of SPR operations transforming one unrooted tree $T_1$ to another $T_2$.
Consider the set of broken edges $E$ of $T_1$ and $T_2$ whose endpoints are modified by applying $S$.
Then we can naturally define the AF underlying $S$ as $T_1 \div E = T_2 \div E$, that is the maximal forest of edges which are not modified by applying $S$.
Note that this AF is not necessarily maximum.
We can similarly define the AF underlying an optimal sequence of TBR moves.

\begin{property}
One or both ends of an edge may move
\label{property:which-end}
\end{property} \noindent

An AF with $k+1$ components represents a set of $k$ TBR moves, each joining two components with respect to $T_2$.
However, SPR moves only move one end of an edge and $S$ may require both ends of a particular edge to move (e.g.\ the edge initially between 3 and 4 in Figure~\ref{fig:spr-distance}).
As such, an optimal SPR path may include one or two rearrangements corresponding to the same broken edge.
In fact, as we show in the next property, optimal SPR paths may require three or more moves corresponding to the same broken edge.

\begin{property}
\label{property:move-twice}
The same endpoint of an edge may move twice
\end{property} \noindent

A useful feature of MAF-like problems is that each optimal move joins two components in the underlying AF.
The minimum distance is thus one less than the number of MAF components, and it is easy to recover an optimal sequence of moves from the~MAF.

However, optimal sequences of uSPR moves with respect to a given underlying agreement forest are not guaranteed to join AF components at each step.
Consider, for example, the pair of trees in Figure~\ref{fig:move-twice-counterexample}.
These trees have only one MAF, but each of their optimal SPR paths begins by applying a move that does not join an underlying AF component.
We verified this by computing the SPR distance between the second tree and each neighbor of the first tree.
Moreover, we exhaustively tested each optimal SPR path underlain by the MAF and found that each such path moves the same endpoint of some edge twice.
In other words, a broken edge may be moved three or more times in an optimal SPR path.
Thus, even given the AF underlying an (unknown) optimal sequence of SPRs, it is not clear how to determine the sequence of SPRs or even their number.
We call this the \emph{AF-move-recovery problem}, and suspect that the uSPR version of the problem may be NP-hard in its own right.
This is in stark contrast to the trivially solved AF-move-recovery problem for MAF-like problems.

\begin{property}
Common clusters are not always maintained
\label{property:break-clusters}
\end{property} \noindent

A \emph{common cluster} is a set of taxa $L \subset X$ from two $X$-trees, $T_1$ and $T_2$, such that $T_i|L$ and $T_i|(X \setminus L)$ are disjoint and connected by a single edge for $i=1,2$.
In MAF-like problems, it is never necessary to move taxa from one side of a common cluster edge to the other.
In fact, these problems can be decomposed into pairs of common clusters which are solved independently~\cite{baroni2006hybrids,linz2011cluster}.
Such cluster decompositions greatly decrease the computational effort required to solve MAF-like problems, as algorithms to do so scale exponentially with the distance computed within a cluster rather than the total distance~\cite{whidden2014supertrees}.

However, as previously shown by Hickey et al.~\cite{hickey2008spr}, there exist pairs of unrooted trees such that every optimal SPR path violates a common cluster.
This lack of independence between clusters is another sign that uSPR differs from MAF-like problems.

\begin{property}
Common paths may be broken
\label{property:break-common}
\end{property} \noindent

In fact, the situation is even worse than identified by Hickey et al.~\cite{hickey2008spr}.
Consider an optimal sequence $M$ of SPR operations transforming one tree, $T_1$, into another, $T_2$, and the underlying agreement forest $F = T_1 \div E_1 = T_2 \div E_2$ for some $E_1$ and $E_2$.
We say that two paths of edges $p_1 \in T_1$ and $p_2 \in T_2$ are \emph{common paths} with respect to $M$ if they connect the same AF components $C_1$ and $C_2$, that is $T_1 \div (E_1 \setminus p_1) = T_2 \div (E_2 \setminus p_2)$.
Not only is it sometimes necessary to move taxa from one side of a cluster edge to another, it may be necessary to break a common path, as shown in Figure~\ref{fig:break-common}.
This is especially surprising, as a later SPR operation must reform the common path.
However, breaking such a common path may free up sets of moves that would be otherwise impossible.

This observation implies that every AF underlying an optimal set of SPR operations may be a strict subforest of another AF, that is, every such underlying AF may not be maximal.
This is in stark contrast to MAF-like problems where every underlying AF is either maximum or (for rooted hybridization number) maximal.

In summary, we believe that these properties make it unlikely that any uSPR MAF formulation is possible.
The first two properties show that even if the correct edges are identified that need to be modified, it does not appear straightforward to find the optimal sequence of modifications.
The second two properties show that we cannot even assume that edges found in each tree will be preserved in an optimal uSPR path.
Thus, we require a different strategy.

\section{Socket Agreement Forests}
\label{sec:sockets}

Recently, we proposed a new type of agreement forest, socket agreement forests (SAFs)~\cite{whidden2016chain} which we summarize here.
Note that SAFs are required only for Observation~\ref{obs:replug-independence} and in Section~\ref{sec:alg:replug-faster}, so this section can be skipped on a first read.
SPR operations on general trees remove and introduce internal nodes, making it difficult to describe equivalence of sets of moves.
SAFs solve this difficulty by including a finite set of predetermined ``sockets'' which are the only nodes that can be involved in SPR operations and are never deleted or introduced.
However, due to this fixed nature, SAFs are unsuitable for determining the SPR distance metric directly.

A \emph{socket forest} is a collection of unrooted trees with special nodes, called \emph{sockets}.
Socket forests have special edges called \emph{connections} that must be between two sockets.
A collection of them is a \emph{connection set}.
Connections are not allowed to connect a socket to itself, although multiple connections to the same socket are allowed.

We will also use the following terminology on socket forests.
A \emph{move} is the replacement of one connection in a connection set for another.
A \emph{replug move} is a move that only changes one socket of a given connection.
An \emph{SPR move} for a given connection set is a replug move that does not create cycles.
We say that a move \emph{breaks} the replaced connection.

The \emph{underlying} forest of a socket forest $F$ is the forest $F^*$ obtained from $F$ by deleting all connections and suppressing all unconnected sockets.
A socket forest $F$ \emph{permits an unrooted tree} $T$ if it is possible to add connections between the sockets of $F$, resolve any multifurcations in some way, and suppress unconnected sockets to obtain $T$.
We call a connected socket forest with such a connection set a \emph{configuration} of $F$ (e.g. a $T$ configuration of $F$).
Moreover, a socket forest $F$ \emph{permits an SPR path} if each intermediate tree along the path is permitted by $F$.
Given two trees $T_1$ and $T_2$, a \emph{socket agreement forest} (SAF) is a socket forest that permits both $T_1$ and $T_2$.
Note that the underlying forest of an SAF is an AF of $T_1$ and $T_2$.

Let $M = m_1, m_2, \ldots, m_k$ be a sequence of moves transforming tree $T_1$ into tree $T_2$ via an SAF $F$.
We can consider the sequence of trees $T_1, t_1, t_2, \ldots, t_k = T_2$ induced by these moves, that is the sequence of tree configurations obtained by applying $M$ to a $T_1$ configuration of $F$ that results in a $T_2$ configuration of $F$.
We thus discuss sockets and connections in the trees, as shorthand for the sockets and connections in the corresponding socket forest configurations.

Socket forests allow us to be precise concerning how connections change during a sequence of moves, because nodes are not deleted or introduced.
Each socket can be separately identified (e.g.\ with a numbering), so any connection can be described irrespective of the other connections in a socket forest.
As with moves on general trees, we consider the deletion and insertion of a connection as simply changing the endpoint of the connection.
As such, the ``new'' connection maintains the same identifier.
Thus, we can identify changes in a connection by the changes in the sockets it connects, again irrespective of the other connections in a socket forest.
This implies a well defined notion of equivalence of moves: two moves are equivalent if they both attach a given endpoint of the same connection to the same socket.
For example, we can uniquely describe a move as changing the second endpoint of connection $c$ to socket $v'$, regardless of the current state of the socket forest.
We say that a move is \emph{valid} for a socket forest configuration if it can be applied to that configuration.
Similarly, a rearranged and/or modified sequence is valid if it is a sequence of valid moves.

Given an AF $F'$ of two trees $T_1$ and $T_2$, we say that an SPR path between $T_1$ and $T_2$ is \emph{optimal with respect to $F'$} if there exists no shorter SPR path between $T_1$ and $T_2$ where each intermediate tree along the path is permitted by $F'$.
Socket agreement forests are a partial analogue of maximum agreement forests: if we can construct a socket agreement forest for $T_1$ and $T_2$, we can be assured of a valid SPR path between $T_1$ and $T_2$ that is optimal with respect to the underlying agreement forest.
However, it is not trivial to calculate the length of the SPR path between trees for a given socket agreement forest, and thus they are only a partial analogue of rooted maximum agreement forests.
We prove the following lemma in~\cite{whidden2016chain}.
\begin{lemma}
	\label{lem:saf}
  Let $F$ be a socket agreement forest of two trees $T_1$ and $T_2$.
	Then there exists an SPR path between $T_1$ and $T_2$ that is permitted by $F$ and optimal with respect to the AF $F^*$ underlying~$F$.
\end{lemma}

We can define an optimal replug path with respect to an AF $F'$ in an analogous manner. It is then straightforward to apply the proof of Lemma~\ref{lem:saf} to replug moves to obtain the following corollary:

\begin{corollary}
	\label{cor:saf}
  Let $F$ be a socket agreement forest of two trees $T_1$ and $T_2$.
	Then there exists a replug path between $T_1$ and $T_2$ that is permitted by $F$ and optimal with respect to the AF $F^*$ underlying~$F$.
\end{corollary}

\section{The replug distance and maximum endpoint agreement forests}

In this section, we introduce the \emph{replug distance}, which lies between the TBR distance and SPR distance, and develop an agreement forest variant, called an \emph{endpoint agreement forest}, for its calculation.
This notion of agreement forest records the position of broken edges for unrooted trees, and thus does not have the pathologies of unrooted SPR agreement forests described in the first section.

The replug distance on trees is defined in terms of the \emph{replug operation}, inspired by replug moves on socket forests.
A replug operation is an SPR operation that does not necessarily result in a tree.
That is, a replug operation again cuts an edge $e = (u,v)$ of a tree $T$, dividing $T$ into $T_u$ and $T_v$, but the new node $v'$ attached to $u$ may be chosen from either $T_u$ or $T_v$.
If $v' \in T_v$ then the replug operation is identical to an SPR operation.
If, however, $v' \in T_u$, then adding the edge $u,v'$ results in a disconnected graph with two components, one of which is cyclic.
We further extend this operation to apply to arbitrary partially-labeled graphs, such as the result of a replug operation on a tree (Figure~\ref{fig:two-replug}).

\begin{figure}[t]
	\hspace*{\stretch{1}}
	\subfigure[\label{fig:two-replug}]{\includegraphics[scale=0.95]{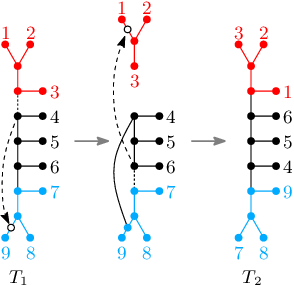}}
	\hspace*{\stretch{2}}
	\subfigure[\label{fig:two-replug-meaf}]{\includegraphics[scale=0.95]{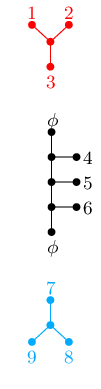}}
	\hspace*{\stretch{2}}
	\subfigure[\label{fig:three-spr-meaf}]{\includegraphics[scale=0.95]{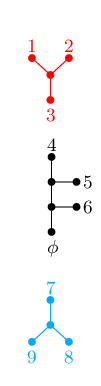}}
	\hspace*{\stretch{1}}

	\caption{(a) Two replug moves transform a tree $T_1$ into another tree $T_2$.
		Replug moves modify one end of an edge, like SPR moves, but may create loops and disconnected components. The dashed line with an arrow is the edge that will be added in the next step; simultaneously the dotted edge touching the tail end of the arrow will be removed.
		(b) A maximum endpoint agreement forest (MEAF) of $T_1$ and $T_2$.
		The $\phi$ nodes indicate endpoints of removed edges that are common to both trees.
		In this case, as there are two components and two $\phi$-nodes, the TBR and replug distances are both 2.
		The SPR distance is 3.
		(c) A MEAF of the trees from Figure~\ref{fig:three-spr}.
		In this case, the replug distance and SPR distance are both 3, while the TBR distance is 2.
	}
	\label{fig:replug}
\end{figure}

Given two trees, $T_1$ and $T_2$, the replug distance, $\dreplug{T_1, T_2}$, is the minimum number of replug moves required to transform $T_1$ into $T_2$.
By dropping the requirement that each intermediate move result in a tree, we achieve an MAF-like problem that can be computed using a structure we call an \emph{endpoint} agreement forest (EAF).
The replug distance can be used as an approximation and lower bound of the SPR distance which, as we discuss in Section~\ref{sec:astar}, can be used in an incremental heuristic search for the exact SPR distance.
We show that it is fixed-parameter tractable in Theorem~\ref{thm:compute-replug}, but conjecture:
\begin{conjecture}
	The replug distance between two trees $T_1$ and $T_2$ is NP-hard to compute.
\end{conjecture}

To develop an agreement forest that represents replug operations, we must consider the properties of the unrooted SPR distance in Section~\ref{sec:properties} which make it difficult to solve using agreement forests.
Property~\ref{property:which-end} states that one or both ends of an edge may move in an optimal unrooted SPR path.
Recall that each SPR operation moves one endpoint of an edge and so any agreement forest formulation must be able to represent which side of an edge remains fixed, as well as the case where both sides of an edge must be moved.
The replug distance also has this directional property.
However, as replug moves are not required to maintain a tree structure, they do not prevent other replug moves in the way that uSPR moves do.
We thus have the following observation:
\begin{obs}
	\label{obs:replug-independence}
	Given an SAF, an endpoint of any connection can be attached to any socket with a valid replug move.
\end{obs}

We can use this observation to show that the replug distance does not have Properties~\ref{property:move-twice},~\ref{property:break-clusters}~and~\ref{property:break-common} of the unrooted SPR distance.
\begin{restatable}{lemma}{lemreplugniceproperties}
	\label{lem:replug-nice-properties}
	Let $M = m_1, m_2, \ldots, m_k$ be an optimal sequence of replug moves transforming a tree $T_1$ into a tree $T_2$.
	Then $M$ does not (1) move the same endpoint $u$ of an edge twice, (2) break a common cluster of $T_1$ and $T_2$, or (3) break an edge of a common path of $T_1$ and $T_2$ with respect to $M$.
\end{restatable}


To account for Property~\ref{property:which-end} and define an agreement forest that represents replug operations, we thus need to represent a fixed endpoint of a moved edge (analogous to component roots of a rooted AF) as well as the case where both endpoints are moved (analogous to an unrooted AF for TBR calculation).
We will call this augmented tree structure a (phylogenetic) \emph{$X$-$\phi$-tree}, which is a generalization of an $X$-tree.
It has $n \ge |X|$ labeled leaves, $|X|$ of which are uniquely labeled from $X$, and the remaining $q(T) := n - |X|$ leaves are each labeled $\phi$.
These $\phi$ nodes will be used to indicate an edge endpoint which remains fixed during a set of tree moves.
As with $X$-trees, two $X$-$\phi$-trees are considered equal if and only if there is an isomorphism between their nodes and edges that maintains node labels.
Thus, $\phi$ nodes are interchangeable.
An $X$-$\phi$-forest is a forest of $X$-$\phi$-trees.

Consider the three possiblities of an edge $e=(u,v)$ of an $X$-$\phi$-tree that is modified by one or more replug moves:
\begin{enumerate}
\item the $v$ endpoint is moved such that $e$ becomes $(u,x)$
\item the $u$ endpoint is moved such that $e$ becomes $(w,v)$
\item both endpoints are moved such that $e$ becomes $(y,z)$.
\end{enumerate}
In order to describe such changes as part of a type of agreement forest, we attach $\phi$ nodes as follows.

In the first case we represent this replug operation by cutting edge $e$, suppressing node $v$, then attaching a new leaf labeled $\phi$ to node $u$.
The second case is similar: we suppress $u$ then add a $\phi$ node to $v$.
Observe that the third case is equivalent to two replug operations or a single TBR operation.
We can thus represent the third case by cutting edge $e$ and suppressing nodes $u$ and $v$, as in general unrooted AFs.

To complement the notion of ``directional'' agreement forest we also need a notion of directional edge set, as follows.
We define an \emph{endpoint edge} $E_c$ of a tree $T$ to be an edge along with a proper subset of its endpoints.
This will be denoted $\eedge{e}{p}$ for an edge $e = (u,v)$ where $p \subsetneq \{u,v\}$.
Note that $p = \emptyset$ indicates that both endpoints are moved.
In the context of a sequence of replug moves, these subsets $p$ are the nodes that remain fixed, which we will call \emph{augmented endpoints}.
An endpoint edge of an $X$-$\phi$-tree cannot have a $\phi$ node as an augmented endpoint.

An endpoint edge set $E$ is a set of endpoint edges $E_c$ of $T$.
For an endpoint edge set $E$, we use $F-E$ to denote the $X$-$\phi$-forest obtained by deleting the edges in $E$ from $F$ and adding $\phi$ node neighbors to each augmented endpoint.
We call this ``cutting'' the endpoint edge set, which is in general a many-to-one mapping.
$F \div E$ is again the ($X$-$\phi$) forest yielded by $F - E$ and we say that $F \div E$ is an endpoint forest of $F$.

An \emph{endpoint agreement forest} (EAF) is now naturally defined in terms of endpoint forests.
Given trees $T_1$ and $T_2$ and forests $F_1$ of $T_1$ and $F_2$ of $T_2$, a forest $F$ is an EAF of $F_1$ and $F_2$ if it is an endpoint forest of both trees.
Observe that an EAF is a generalization of an AF, where either no nodes remain fixed (unrooted AF, in which every $p = \emptyset$), or every endpoint furthest from the root remains fixed (rooted AF, in which no $p = \emptyset$).
As such, we can always find an EAF for two trees by constructing an AF.

The \emph{weight} of an EAF $F$ is defined as:
$$\weight(F) = 2\left(\size{F}-1\right) - q(F),$$
where $q(F)$ is the number of $\phi$ nodes in $F$.
Observe that the weight strictly increases upon cutting an endpoint edge set.

We say that an EAF $F$ of two trees $T_1$ and $T_2$ is a \emph{maximum endpoint agreement forest} (MEAF) of $T_1$ and $T_2$ if it has minimum weight (Figures~\ref{fig:two-replug-meaf} and~\ref{fig:three-spr-meaf}).
Use $\weight(T_1, T_2)$ to denote this minimum weight.
We can show that $\dreplug{T_1, T_2}$ is equal to $\weight(T_1, T_2)$, rather than the number of components.

\begin{restatable}{theorem}{thmmeaf}
\label{thm:meaf}
Let $T_1$ and $T_2$ be unrooted trees.
Then $\dreplug{T_1,T_2} = \weight(T_1, T_2)$.
\end{restatable}

Our proof of Theorem~\ref{thm:meaf} provides an inductive procedure for constructing an optimal sequence of replug moves from an MEAF (Fig.~\ref{fig:meaf} in the appendix).
Moreover, each step of this procedure can be implemented to require linear time using the tree to AF mappings we construct later with Lemma~\ref{lem:map-edges}.
Thus:

\begin{restatable}{corollary}{coreaftoreplug}
\label{cor:eaftoreplug}
Let $F$ be an EAF of two unrooted trees $T_1$ and $T_2$.
	A sequence of $\weight(F)$ replug moves that transform $T_1$ into $T_2$ can be obtained from $F$ in $\OhOf{n\weight(F)}$-time.
\end{restatable}

The replug distance is a lower bound for the SPR distance, which enables the fast SPR distance algorithm in Section~\ref{sec:astar}.

\begin{restatable}{theorem}{thmreplugbound}
\label{thm:replug-bound}
For any pair of trees $T_1$ and $T_2$,
$$\dtbr{T_1, T_2} \le \dreplug{T_1, T_2} \le \dspr{T_1, T_2}.$$
\end{restatable}

\section{A fixed-parameter replug distance algorithm}
\label{sec:algorithm}

In this section we develop a fixed-parameter algorithm for computing the replug distance between a pair of binary rooted trees $T_1$ and $T_2$.
We apply a two-phase bounded search tree approach to determine whether the replug distance is at most a given value $k$.
The minimum replug distance can be found by repeatedly running this algorithm with increasing parameters $k=1,2, \ldots, \dreplug{T_1, T_2}$.

Our strategy is to enumerate all possible maximal agreement forests (Section~\ref{sec:alg:tbr}) and then decorate them with $\phi$ nodes to enumerate the possible maximal endpoint agreement forests (Section~\ref{sec:alg:replug}).
We then improve the second phase by considering each possible socket assignment as a SAF of $F$ and finding an optimal $\phi$ node assignment (Section~\ref{sec:alg:replug-faster}).

\subsection{Enumerate maximal AFs (TBR distance)}
\label{sec:alg:tbr}

The first phase of our search enumerates all maximal (unrooted) AFs.
We modify the $\OhOf{4^k n}$-time MAF (TBR) distance algorithm of Whidden and Zeh~\cite{whidden2009unifying} to do so, and adopt some of the improved branching cases from the $\OhOf{3^k n}$-time algorithm of Chen et al.~\cite{chen2013parameterized}.
Neither algorithm is capable of enumerating all mAFs in its original form and, furthermore, we can not apply all of the cases from Chen et al.\  because some of their cases necessarily miss some mAFs.

Given two trees $T_1$ and $T_2$ and a parameter $k$, our algorithm $\alg{T_1, T_2, k}$ finds each maximal AF that can be obtained by cutting $k$ or fewer edges.
If no such mAF exists, then it returns an empty set.
If mAFs are found, we apply the next phase of our search to each of these candidate underlying AFs.
If we find no mAFs or if none of these mAFs underly a maximum endpoint agreement forest, we increase $k$ and repeat until we find one.
This does not increase the running time of the algorithm by more than a constant factor because the running time depends exponentially on $k$.

Our algorithm is recursive, acting on the current forests of $T_1$ and $T_2$ that are obtained by cutting edge sets $E_1$ and $E_2$, respectively, with $\size{E_1} = \size{E_2}$.
These forests maintain a specific structure, $F_1 = T_1 \div E_1 = \{\dot{T}_1\} \cup F_0$, and $F_2 = T_2 \div E_2 = \dot{F}_2 \cup F_0$, defined as follows:
\begin{itemize}
\item The tree $\dot{T}_1$ is obtained from $T_1$ by cutting $E_1$.
\item The forest~$F_0$ has all of the rooted subtrees cut off from $T_1$. The rooting for these subtrees is that given by rooting at the node touching a cut edge.
\item The forest $\dot{F}_2$ has the same label set as $\dot{T}_1$, and contains the components of $F_2$ that do not yet agree with $\dot{T}_1$ (i.e. $(F_2 \div E_2) \setminus F_0$).
\item A set $R_t$ (roots-todo) stores the roots of (not necessarily maximal) subtrees agreeing between $\dot{T}_1$ and~$\dot{F}_2$.
\end{itemize}

Each invocation of the algorithm works to modify $\dot{T}_1$ and $\dot{F}_2$ towards agreement, collecting the needed agreement forest components in $F_0$.
For the top-level invocation, $\dot{T}_1 = T_1$, $\dot{F}_2 = \set{T_2}$, $F_0 = \emptyset$, and $R_t$ contains all leaves of $T_1$.

We say that a pair of nodes $a,c \in R_t$ that are siblings in $\dot{T}_1$ are a \emph{sibling pair} $\set{a,c}$.
Such a pair must exist when $\size{R_t} \ge 2$ just as every tree must have a pair of sibling leaves.
For each node $x \in R_t$, $e_x$ denotes the unique edge of $\dot{T}_1$ or $\dot{F}_2$ that is adjacent to $x$ and extends out of the common subtree rooted at $x$.
We call this the \emph{edge adjacent to $x$}.

\begin{figure*}[t]
	\hspace*{\stretch{1}}
	\subfigure[\label{fig:tbr-cases-t1}]{\includegraphics[scale=0.6]{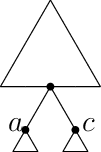}}
	\hspace*{\stretch{2}}
	\subfigure[\label{fig:tbr-cases-sc}]{\includegraphics[scale=0.6]{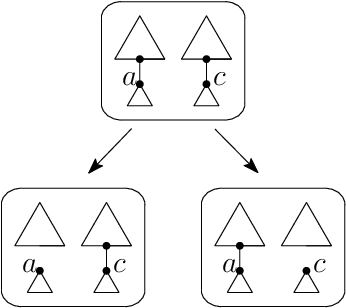}}
	\hspace*{\stretch{1}}
	\subfigure[\label{fig:tbr-cases-cab}]{\includegraphics[scale=0.6]{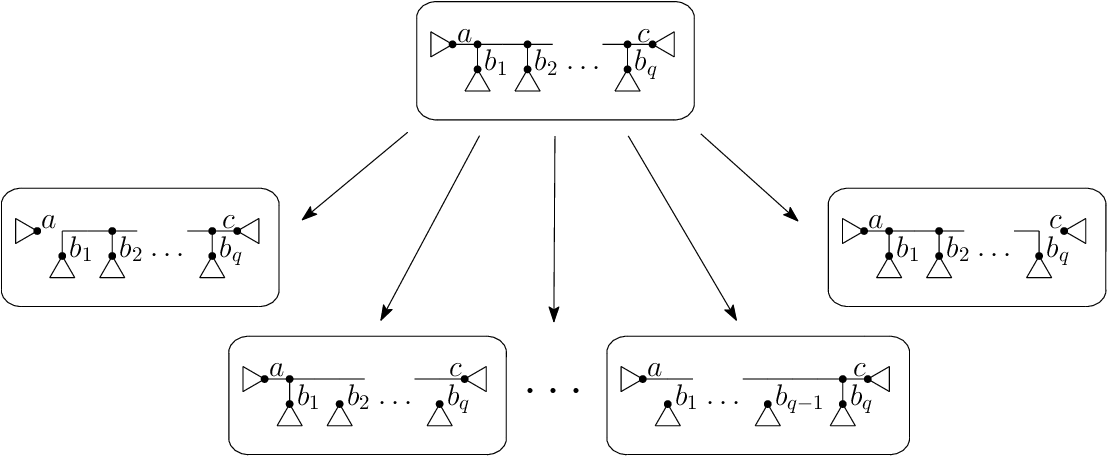}}
	\hspace*{\stretch{1}}

	\caption{Branching cases of the maximal AF enumeration algorithm.
		(a)~A~sibling pair $\set{a,c}$ of $F_1$.
		(b)~The branching cases applied to $F_2$ when $a \noreach{F_2} c$.
		(c)~The branching cases applied to $F_2$ when $a \reach{F_2} c$.
	}
	\label{fig:tbr:cases}
\end{figure*}

\vspace{1em}
\noindent\textbf{Algorithm $\alg{\bm{F_1, F_2, k}}$:}
\vspace{-0.2em}

\begin{enumerate}[label=\arabic{*}.,ref=\arabic{*},leftmargin=*]
\item \label{case:abort} (Failure)
	If $k < 0$, then return $\emptyset$.
\item \label{case:success} (Success)
	If $|R_t| = 1$, then $F_1 = F_2$.
	Hence, $F_2$ is an AF of $T_1$ and $T_2$.
  Return \set{$F_2}$.
\item \label{case:tbr:singleton} (Prune maximal agreeing subtrees)
  If there is no node $r \in R_t$ that is a root in $\dot{F}_2$, proceed to Step~\ref{item:choose-sib-pair}.
  Otherwise choose such a node $r \in R_t$; remove it from~$R_t$ and move the subtree $S$ of $\dot{F}_2$ containing $r$ to $F_0$; cut the edge $\edge{r}$ separating $S$ in $\dot{T}_1$ from the rest of $\dot{T}_1$, which gives $S$ a rooting.
  Return to Step~\ref{case:success}.
\item \label{item:choose-sib-pair} Choose a sibling pair $\set{a, c}$ of $\dot{T}_1$. Let $\edge{a}$ and $\edge{c}$ be the edges adjacent to $a$ and $c$ in $\dot{F}_2$.
\item \label{case:tbr:sibling} (Grow agreeing subtrees)
  If $a$ and $c$ are siblings in $\dot{F}_2$, do the following:
	remove $a$ and $c$ from $R_t$; label their mutual neighbor in both forests with $(a, c)$ and add it to $R_t$; return to Step~\ref{case:success}.
\item \label{case:tbr:non-sibling} (Branching) Distinguish two cases depending on whether $a \reach{F_2} c$ (Figure~\ref{fig:tbr:cases}).
Note that because $a$ and $c$ are in $\dot{T}_1$ they are not in $F_0$, so $a \reach{F_2} c$ is equivalent to $a \reach{\dot{F}_2} c$.
  \begin{enumerate}[label=\theenumi.\arabic{*}.,leftmargin=*,ref=\theenumi.\arabic{*}]
		\item \label{case:tbr:is} If $a \noreach{F_2} c$, try cutting off either the subtree rooted at $a$ or that rooted at $c$, returning:
	$$\alg{F_1, F_2 \div \set{\edge{a}}, k-1} \cup \alg{F_1, F_2 \div \set{\edge{c}}, k-1}.$$
\item \label{case:tbr:mpe} If $a \reach{F_2} c$, let $b_1, b_2, \ldots, b_q$ be the pendant nodes on the path from $a$ to $c$ in $T_2$ and let $\edge{b_i}$ be the pendant edge adjacent to $b_i$.
	Try cutting $\edge{a}$ or $\edge{c}$, or try cutting off all but one of the subtrees on the path between $a$ and $c$, returning:
    \begin{align*}
	&\alg{F_1, F_2 \div \set{\edge{a}}, k-1} \\
	\cup \ &\alg{F_1, F_2 \div \set{\edge{b_2}, \edge{b_3}, \ldots, \edge{b_q}}, k-(q-1)}\\
	 \cup \ &\alg{F_1, F_2 \div \set{\edge{b_1}, \edge{b_3}, \edge{b_4}, \ldots, \edge{b_q}}, k-(q-1)}\\
	 \ \ldots \\
	 \cup \ &\alg{F_1, F_2 \div \set{\edge{b_1}, \edge{b_2}, \ldots, \edge{b_{q-1}}}, k-(q-1)}\\
	 \cup \ &\alg{F_1, F_2 \div \set{\edge{c}}, k-1}.
    \end{align*}
	\end{enumerate}

\end{enumerate}

The correctness of our algorithm is proven using a simplification of the proof of Theorem 2 of Whidden and Zeh~\cite{whidden2009unifying}, modified for arbitrary maximal agreement forests.

\begin{restatable}{lemma}{lemtbr}
\label{lem:tbr}
$\alg{T_1, T_2, k}$ returns the set of maximal agreement forests of $T_1$ and $T_2$ that can be obtained by cutting $k$ or fewer edges.
\end{restatable}

\begin{restatable}{theorem}{thmtbr}
\label{thm:tbr}
Given two trees $T_1$ and $T_2$, the set $S$ of maximal agreement forests of $T_1$ and $T_2$:
\begin{enumerate}
\item can be enumerated in $\OhOf{4^k n}$-time,
\item can be returned in $\OhOf{4^k n + \size{S}\log{\size{S}}}$-time (the $\log{\size{S}}$ factor required to avoid duplicate agreement forests), and
\item is of size at most $4^k$.
\end{enumerate}
\end{restatable}

\subsection{Enumerate maximal EAFs (Replug distance)}
\label{sec:alg:replug}

The second phase of our algorithm enumerates maximal endpoint agreement forests to compute the replug distance.
Using $\alg{T_1, T_2, k}$ we can enumerate each mAF $F$ of two trees $T_1$ and $T_2$ that has $k+1$ or fewer components, in other words, the mAFs that correspond to $k$ or fewer TBR operations.
However, such an mAF may imply anywhere from $k$ to $2k$ SPR operations.
We assign $\phi$ nodes to $F$ to determine the maximal EAF (mEAF) with minimum weight.

We require a method to map nodes of the trees to the agreement forest and vice versa.
Given an mAF $F$ of trees $T_1$ and $T_2$, we construct a ``forward'' mapping $\psi(x)$ that maps nodes of $T_1$ and $T_2$ to nodes of $F$, as well as a ``reverse'' mapping $\psi^{-1}(T_i, x)$, $i \in \set{1,2}$, that maps nodes of $F$ to the trees.
Recall that $F = T_1 \div E_1 = T_2 \div E_2$, where $E_1$ and $E_2$ are the edges cut from $T_1$ and $T_2$, respectively.
As $F$ contains fewer nodes than $T_1$ and $T_2$, our forward mapping $\psi(x)$ maps some nodes of $T_1$ and $T_2$ to the empty set.
These are exactly the nodes that are contracted when $E_1$ and $E_2$ are cut to obtain $F$ from $T_1$ and $T_2$.
We call these the \emph{dead} nodes of the trees, in contrast to the \emph{alive} nodes that are mapped to nodes of $F$.

There are potentially many edge sets that can be cut from the trees to obtain the mAF $F$, but we can quickly obtain one pair of such edge sets and their induced mappings:

\begin{restatable}{lemma}{lemmapedges}
	\label{lem:map-edges}
	Let $F$ be an mAF of two trees $T_1$ and $T_2$. Then
	\begin{enumerate}
	\item two sets of edges $E_1$ and $E_2$ such that $T_1 \div E_1 = T_2 \div E_2 = F$ can be constructed in linear time, and
	\item a mapping $\psi(x)$ from $T_1$ and $T_2$ nodes to nodes of $F$ and the reverse mappings $\psi^{-1}(T_i, x)$ can be built in $\OhOf{n}$-time; lookups using these mappings take constant time.
	\end{enumerate}
\end{restatable}

Now, we develop an algorithm $\replug{T_1, T_2, F, d}$ to enumerate the set of mEAFs of weight $d$ or less with the same components (modulo $\phi$ nodes) as a given mAF $F$ of trees $T_1$ and $T_2$.
First assume that the edge set $E_1$ to remove from $T_1$ is pre-specified.
In this case, each corresponding mEAF $F'$ is simply $F$ with a choice of endpoint for each edge $e = (u,v) \in E_1$, that is, $\eedge{e}{\set{u}}$, $\eedge{e}{\set{v}}$, or $\eedge{e}{\emptyset}$.
We have three choices for each edge, and can therefore enumerate these candidate mEAFs in $\OhOf{3^k n}$ time, where $k+1$ is the number of components of $F$, and keep any of weight $d$ or less.
Note that, for each candidate mEAF, we must also check, using $\OhOf{n}$-time, whether it is a forest of $T_2$ (it is guaranteed to be a forest of $T_1$ as we enumerated choices of endpoints for $E_1$).
We define the following recursive subprocedure to enumerate the candidate mEAFs that can be obtained from an mAF $F$ with a given set of edges $E_1$.

\vspace{1em}
\noindent\textbf{Algorithm $\replugdecorate{\bm{T_1, T_2, F, E_1, k}}$:}
\vspace{-0.2em}
\begin{enumerate}[label=\arabic{*}.,ref=\arabic{*},leftmargin=*]
\item \label{case:replugdecorate:abort} (Failure)
	If $k < 0$ then return $\emptyset$.
\item \label{case:replugdecorate:success} (Success)
	If $F$ contains an unprocessed edge $e = (u,v) \in E_1$, then proceed.
	Otherwise $F$ is a forest of $T_1$ with weight at most double the initial value of $k$.
	If $F$ is also a forest of $T_2$, return \set{$F_2}$.
	Otherwise, return $\emptyset$.
\item \label{case:replugdecorate:branch} (Branch)
	Return: \\
	$\replugdecorate{T_1, T_2, F \div \eedge{e}{\set{u}}, E_1 \setminus \set{e}, k-1} \cup \\
	\replugdecorate{T_1, T_2, F \div \eedge{e}{\set{v}}, E_1 \setminus \set{e}, k-1} \cup \\
	\replugdecorate{T_1, T_2, F \div \eedge{e}{\emptyset}, E_1 \setminus \set{e}, k-1}$.
\end{enumerate}

Define a \emph{dead tree} to be a tree induced by some maximal set of adjacent edges in $T_1$ or $T_2$ that do not map to edges of $F$. 
When every dead tree is a single edge, the edge sets $E_1$ and $E_2$ removed to obtain $F$ from $T_1$ and $T_2$ are unique and thus known.
However, in general we must enumerate every set of edges $E_1$ such that $T_1 \div E_1 = F$, because a component $C$ of $F$ that is adjacent to a dead tree with three or more leaves in both $T_1$ and $T_2$ may have nontrivial $\phi$ node structures in the mEAF.
Given the node mappings and reverse node mappings constructed with Lemma~\ref{lem:map-edges}, we can identify the dead trees in a tree with three preorder traversals.
The node mapping construction procedure roots the trees arbitrarily.
We use the same roots in our traversals.

The first traversal determines the depth from each node of $T_1$ to its arbitrary root.
The second traversal identifies the path induced by each edge $e=(u,v)$ of $F$ in $T_1$.
To do so, we move towards the root from $\psi^{-1}(T_1, u)$ and $\psi^{-1}(T_1, v)$, always moving from the more distant node.
These paths will cross at the least common ancestor of $\psi^{-1}(T_1, u)$ and $\psi^{-1}(T_1, v)$ with respect to the root.
We mark $\psi^{-1}(T_1, u)$ and $\psi^{-1}(T_1, v)$ as alive and each other node of the path as a socket.
If $u$ or $v$ is a component root of $F$, then it was not removed by a forced contraction when its adjacent edge was cut.
In this case we also mark $\psi^{-1}(T_1, u)$ or $\psi^{-1}(T_1, v)$, respectively, as a socket.
The third and final traversal marks all the remaining nodes as dead and identifies the connected components of dead trees.
We will use this assignment of dead nodes and socket nodes at the end of this section to quickly find a single optimal $\phi$ node assignment.

A dead tree with $q$ leaves has $2q-3$ edges.
Each such tree can be removed by removing $q-1$ edges and contracting the resulting degree two nodes.
Thus there are ${2q-3 \choose q-1} \le 4^q$ choices of edges to cut that result in the same dead tree.
This implies that there are at most $4^k$ edge sets that result in the same AF.
We thereby construct each possible $E_1$ set and test each combination of three choices per edge, as in the unique edge set case.
This requires $\OhOf{4^k 3^k n} = \OhOf{12^k n}$ time.

In summary, the high-level steps of the algorithm to calculate the replug distance are as follows.

\vspace{1em}
\noindent\textbf{Algorithm $\replug{\bm{T_1, T_2, F, d}}$:}
\vspace{-0.2em}

\begin{enumerate}[label=\arabic{*}.,ref=\arabic{*},leftmargin=*]
	\item Root $T_1$ and $T_2$ arbitrarily at nodes $r_1$ and $r_2$.
	\item Construct the mappings $\psi$ and $\psi^{-1}$ using Lemma~\ref{lem:map-edges}.
	\item Compute the distance from each node $n$ in $T_1$ to $r_1$.
	\item For each edge $e = (u,v)$ of $F$, mark $\psi^{-1}(T_1, u)$ and $\psi^{-1}(T_1, v)$ as alive and the other nodes on the path in $T_1$ from $\psi^{-1}(T_1, u)$ to $\psi^{-1}(T_1, v)$ as sockets.
	\item Mark all unmarked nodes dead.
	\item Identify the dead components $D_1 = \set{d_1, d_2, \ldots d_{q_1}}$ of $T_1$.
	\item Repeat the previous steps for $T_2$, identifying the dead components $D_2$.
	\item Let $\mathcal{F} \gets \emptyset$.
	\item For each edge set $E_1$ induced by $D_1$,
		\begin{enumerate}
			\item Let $k \gets \size{E_1}$.
			\item Let $\mathcal{F}' \gets \replugdecorate{T_1, T_2, F, E_1, k}$.
			\item For each forest $F' \in \mathcal{F}'$, if $F'$ has weight at most $d$ and is consistent with some edge set $E_2$ induced by $D_2$ then add it to $\mathcal{F}$.
		\end{enumerate}
		\item Return $\mathcal{F}$.
\end{enumerate}

In combination with the algorithm $\alg{T_1, T_2, k}$ from Section~\ref{sec:alg:tbr}, this implies:

\begin{restatable}{theorem}{thmenumeratemeafs}
		\label{thm:enumerate-meafs}
	Given two trees $T_1$ and $T_2$, the set of mEAFs of $T_1$ and $T_2$ with $k+1$ or fewer components can be enumerated in $\OhOf{48^k n}$-time and there are at most $48^k$ such mEAFs.
\end{restatable}

This is a fairly loose bound as there are typically far fewer than $4^k$ mAFs~\cite{voorkamp2014maximal} and the majority of agreement forests will have few and small dead trees.
Although the above procedure requires more effort to enumerate the mEAFs that can be obtained from a single mAF in the worst case than it does to enumerate each of the mAFs, we expect the opposite to be the typical case in practice because closely related trees (such as those compared in phylogenetic analysis) rarely differ in only a single large connected set of bipartitions in both trees which is necessary to induce a single large dead component in both trees.

We can improve the $\phi$ node assignment step in the case that $k$ is close to the SPR distance.
If the SPR distance is $d$ then it never makes sense to consider an mEAF with weight exceeding $d$ for the final phase of our search.
In fact, we can exclude any mEAF $F'$ with fewer than $d-k$ $\phi$ nodes, as their weight (and thus replug distance) is guaranteed to exceed $d$.
We thus define $\replugdecoratetwo{T_1, T_2, F', E_1, d}$, a version of $\textsf{replug-decorate}$ which takes this observation into account to enumerate a more limited set of candidate mEAFs.
For the definition, initially $F'=T_1$ and $E_1$ contains each edge of $F$, as follows:

\vspace{1em}
\noindent\textbf{Algorithm $\replugdecoratetwo{\bm{T_1, T_2, F', E_1, d}}$:}
\vspace{-0.2em}
\begin{enumerate}[label=\arabic{*}.,ref=\arabic{*},leftmargin=*]
\item \label{case:replugdecorate2:abort} (Failure)
	If $d < 0$ then return $\emptyset$.
\item \label{case:replugdecorate2:success} (Success)
	If $F'$ contains an unprocessed edge $e = (u,v) \in E_1$, then proceed.
	Otherwise $F'$ is a forest of $T_1$ with weight less than the initial value of $d$.
	If $F'$ is also a forest of $T_2$, return \set{$F_2}$.
	Otherwise, return $\emptyset$.
\item \label{case:replugdecorate2:branch} (Branch)
	Return: \\
	$\replugdecoratetwo{T_1, T_2, F' \div \eedge{e}{\set{u}}, E_1 \setminus \set{e}, d-1} \cup \\
	\replugdecoratetwo{T_1, T_2, F' \div \eedge{e}{\set{v}}, E_1 \setminus \set{e}, d-1} \cup \\
	\replugdecoratetwo{T_1, T_2, F' \div \eedge{e}{\emptyset}, E_1 \setminus \set{e}, d-2}$.
\end{enumerate}

\begin{restatable}{theorem}{thmenumeratemeafsd}
		\label{thm:enumerate-meafs-d}
	Given two trees $T_1$ and $T_2$, the set of mEAFs of $T_1$ and $T_2$ with weight at most $d$ and $k+1$ or fewer components can be enumerated in $\OhOf{\min(4^k, 2.42^d) \cdot 12^k n}$-time and there are at most $(\min(4^k, 2.42^d) \cdot 12^k)$ such mEAFs.
\end{restatable}

\subsection{Quickly finding a single MEAF (Replug distance)}
\label{sec:alg:replug-faster}

Finally, we discuss how to quickly compute the replug distance in practice.
Our goal now is to determine whether the replug distance between two trees is at most $d$.
The minimum distance can again be found by testing increasing values of $d$ starting from 0.
We again enumerate the maximal agreement forests that can be obtained by removing at most $d$ edges with $\alg{T_1, T_2, d}$.
For each such mAF $F$, however, we now wish to find a single optimal $\phi$ node assignment, rather than each such optimal~assignment.

To determine such an optimal assignment, we again start by finding an initial set of edges $E_1$ and $E_2$ such that $T_1 \div E_1 = T_1 \div E_2 = F$.
Our first optimization stems from the observation that nontrivial $\phi$ node structures are only possible when a component $C$ of an mAF is adjacent to a dead tree in both trees $T_1$ and $T_2$.
Call a dead tree \emph{uncertain} if it is adjacent to such a component in a given tree.
When enumerating edge sets $E'_1$ and $E'_2$ that can be used to obtain $F$, we only need to enumerate combinations of dead tree edges from dead trees that are uncertain in both $T_1$ and~$T_2$.

Now, given an AF $F$ of $T_1$ and $T_2$ and edge sets $E_1$ and $E_2$ that can be used to obtain $F$ from the trees, we develop our optimal $\phi$ node assignment procedure.
The intermediate structures are illustrated in Fig.~\ref{fig:dead-tree-eaf} in the appendix.
We first identify the set of candidate agreement forest edges that can be adjacent to a $\phi$-node (\emph{candidate $\phi$-nodes}).
We next identify the set of constraints on EAFs induced by these candidates and show that such sets belong to the monotone class $\cnfletwo$~\cite{porschen2007algorithms}.
This is the class of satisfiability formulas in conjuctive normal form with no negated literals in which each literal occurs in at most 2 clauses.
As such, we reduce the problem of decorating an AF with $d-k$ $\phi$ nodes, to that of determining whether a boolean $\cnfletwo$ formula can be satisfied by an assignment with at least $d-k$ variables set to true.
Finally, we briefly describe how this latter problem can be solved in $\OhOf{k^{1.5}}$-time by finding a minimum edge cover in an equivalent clause graph~\cite{porschen2007algorithms}.

To identify the candidate $\phi$-nodes, we convert the AF $F$ to an SAF $S$.
Recall that a $\phi$-node indicates an edge endpoint which remains fixed during a set of moves transforming one tree, $T_1$, into another, $T_2$.
This implies that one socket remains fixed for each $\phi$-node decorated edge in the $T_1$ and $T_2$ configurations of $S$.
Thus, the sockets of $S$ which are connected to an edge in both $T_1$ and $T_2$ are exactly our set of candidate $\phi$-nodes.

To accomplish this SAF conversion, we must thus introduce sockets from both $T_1$ and $T_2$ and match sockets which represent nodes from both trees.
As described in the \textsf{replug} algorithm, a path of dead nodes in $T_1$ (respectively, $T_2$) between two alive nodes form a set of $T_1$ ($T_2$) sockets.
We say that these dead nodes are all on the same edge of $F$.
These sets are ambiguous in the sense that any pair consisting of a $T_1$ and $T_2$ socket that are on the same edge may be a fixed endpoint that can have a $\phi$ node.
We can identify these ambiguous sets, as well as the dead nodes which are not sockets, in linear time as previously discussed.
We must choose a mapping for each such set of $h_1$ $T_1$ sockets and $h_2$ $T_2$ sockets that map to the same alive node edge and that maintains the same orientation.
There are thus $h_1 \choose h_2$ choices for each such set (assuming $h_1 \ge h_2$), for a maximum of at most $4^{h_1} \le 4^k$ combinations that must be considered.
We further observe that dead tree nodes are not possible choices in such a mapping.
Therefore, there are at most $4^k$ SAFs that must be considered, stemming from both choices of $E_1$ and $E_2$ and socket mappings.
We expect such situations to be degenerate in practice and that most cases will induce a small number of SAFs based on our experience with rooted agreement forests.

Now, for each SAF, we identify the set of constraints induced by the candidate $\phi$-nodes.
These constraints naturally arise from the fact that we cannot assign two $\phi$-nodes to both endpoints of the same cut edge in $E_1$ of $T_1$ or $E_2$ of $T_2$, as that would imply that the edge remains fixed.
Now, suppose, without lack of generality, we have a dead tree $D$ adjacent to $n_D$ sockets of $F$ in $T_1$.
We say that an assignment of $\phi$-nodes to sockets of $S$ \emph{satisfies} $D$ if at least one of the $n_D$ sockets adjacent to $D$ is not assigned a $\phi$ node.
By Lemma~\ref{lem:dead-tree-assignment}, it suffices to find an assignment of $\phi$ nodes to $S$ that satisfy each dead component of $T_1$ and $T_2$.

\begin{restatable}{lemma}{lemdeadtreeassignment}
\label{lem:dead-tree-assignment}
	Given a socket agreement forest $S$ of two trees $T_1$ and $T_2$, an assignment of $\phi$-nodes to sockets of $S$ is an EAF of $T_1$ and $T_2$ if, and only if, the assignment satisfies every dead tree in $T_1$ and $T_2$ with respect to $S$.
\end{restatable}

We now show that the full set of such constraints is a boolean monotone $\cnfletwo$ formula~\cite{porschen2007algorithms}.
In particular, determining the maximum number of $\phi$ nodes that can be added to an SAF is thus equivalent to determining the minimum number of variables which must be true in such a formula---the minimum cardinality satisfiability problem.

\begin{restatable}{lemma}{lemreplugcnfletwo}
	\label{lem:replug-cnfletwo}
	The replug distance problem on an SAF $S$ of trees $T_1$ and $T_2$ can be solved by solving the minimum cardinality satisfiability problem on a boolean monotone $\cnfletwo$ formula.
\end{restatable}

This satisfiability problem can be solved efficiently with a polynomial time algorithm.
Unlike general CNF satisfiability, which is NP-hard~\cite{porschen2007algorithms}, the minimum cardinality satisfiability problem on monotone $\cnfletwo$ formulas is equivalent to the edge covering problem on a \emph{clause graph}.

\begin{restatable}{lemma}{lemsolveconstraints}
\label{lem:solve-constraints}
Given an SAF $F$ of two trees $T_1$ and $T_2$, an EAF $F'$ of $T_1$ and $T_2$ can be computed in $\OhOf{k^{1.5}}$-time such that $F'$ has $F$ as its underlying agreement forest and $F'$ contains as many $\phi$-nodes as any such EAF.
\end{restatable}

Finally, we combine the fast $\phi$-node assignment procedure of Lemma~\ref{lem:solve-constraints} with our general mAF enumeration to solve the replug distance for a pair of trees:

\begin{restatable}{theorem}{thmcomputereplug}
	\label{thm:compute-replug}
        Given two trees $T_1$ and $T_2$, an EAF $F$ of $T_1$ and $T_2$ with $\weight(F) = \dreplug{T_1, T_2}$ and $k+1$ components can be found (or determined not to exist) in:
        \begin{enumerate}
                \item $\OhOf{4^k (4^k k^{1.5} + n)}$-time, or
                \item $\OhOf{4^k n + Y k^{1.5}}$-time,
									where $Y$ is the number of candidate SAFs with at most $k+1$ components.
        \end{enumerate}
\end{restatable}

Repeated applications of the subtree and chain reduction rules reduce the size of the initial trees to a linear function with respect to their SPR distance~\cite{whidden2016chain}.
We thus achieve:

\begin{restatable}{corollary}{cortimereplug}
\label{cor:time:replug}
	The replug distance for a pair of trees $T_1$ and $T_2$ can be solved in $\OhOf{16^d d^{1.5}}$-time, where $d=\dreplug{T_1, T_2}$.
\end{restatable}

\section{Computing the uSPR distance with a progressive A* approach}
\label{sec:astar}

We now present our fixed-parameter algorithm for computing the subtree prune-and-regraft distance between two unrooted trees, $T_1$ and $T_2$.
We first reduce the trees repeatedly by applying the subtree and chain reduction rules.
After these reductions, the trees contain at most $28\dspr{T_1, T_2}$ leaves~\cite{whidden2016chain}.
We then apply an incremental heuristic A* search~\cite{hart1968formal} beginning from $T_1$ to find $T_2$.
We could simply use standard A* with a replug distance heuristic.
However, the replug distance is relatively expensive to compute, and such an algorithm would be very slow.
Instead, we use a sequence of increasingly more accurate but expensive to compute heuristics in a method we call \emph{progressive A* search}.
The benefit of our algorithm over a standard A* search lies in the use of multiple heuristic functions, each of which provides a lower bound on its successor and is significantly less expensive to compute.
This search focuses on paths that likely lead to the target tree while avoiding expensive computation of the TBR and replug distances on each of the $\OhOf{n^2}$ neighbors of every visited tree.

We maintain a priority queue $P$ of trees that remain to be explored.
We also maintain a set $V$ of trees that have already been visited, along with their estimated distances to $T_2$ and what function was used to make that estimate.
Initially $P = \set{T_1}$.
Define the priority $p(T) = (d_T, h_T, e_T)$.
First, $d_T$ is the distance already traveled: $d_T = \dspr{T_1, T}$.
Second, $h_T$ is the estimated minimum distance from $T_1$ to $T_2$ that can be achieved by a path that visits tree $T$.
This estimate $h_T$ is $d_T + e_T(T, T_2)$, where $e_T(T, T_2)$ is the estimated minimum distance from $T$ to $T_2$ using the estimation function $e_T$.
Third, $e_T$ is one of the four estimation functions (i.e., the heuristics) used as follows.
The heuristic $\one(T, T_2) = 1$ for any tree.
The heuristic $\atbr{T, T_2}$ is the linear-time $3$-approximation algorithm for the TBR distance of Whidden and Zeh~\cite{whidden2009unifying} divided by 3 to guarantee a lower bound estimate of the TBR distance.
$\dtbr{T, T_2}$ and $\dreplug{T, T_2}$ are computed with the algorithms in the preceding section.
We impose a total order on these heuristics as described below.

Our search procedure always considers the next tree from $P$ with smallest priority according to the partial ordering of these values where for any two trees $T_i$ and $T_j$, $p(T_i)~<~p(T_j)$ iff:
\begin{enumerate}
\item $h_{T_i} < h_{T_j}$,
\item $h_{T_i} = h_{T_j}$ and $e_{T_i} < e_{T_j}$, or
\item $h_{T_i} = h_{T_j}$, $e_{T_i} = e_{T_j}$, and $d_{T_i} > d_{T_j}$.
\end{enumerate}
In other words, we prioritize the tree $T_i$ with smallest heuristic distance $h_{T_i}$.
We break $h_T$ ties using a total ordering of our heuristics $e_T$:
$\one() < \atbr{} < \dtbr{} < \dreplug{}$.
In turn, $e_T$ ties are broken by partial distances $d_T$.

Each estimator provides a lower bound on each of the successive estimators as well as the target distance $\dspr{T_i, T_2}$, which is an important condition to ensure the correctness of our progressive A* search.
We break ties by selecting the tree that is most distant from the starting position and therefore estimated to be closer to the destination tree.
Trees with equal $h_T$, $e_T$, and $d_T$ values are selected from uniformly at random.

Initially, $p(T_1) = (0, 1, \one())$ when inserting $T_1$ into $P$.
We repeatedly remove the tree $T$ from $P$ with smallest priority $p(T) = (d_T, h_T, e_T)$ and apply one of the options:
\begin{enumerate}
\item if $e_T = \one()$, reinsert $T$ with priority \\$(d_T, d_T + \atbr{T, T_2}, \atbr{})$.
\item if $e_T = \atbr{}$, reinsert $T$ with priority \\$(d_T, d_T + \dtbr{T, T_2}, \dtbr{})$.
\item if $e_T = \dtbr{}$, reinsert $T$ with priority \\$(d_T, d_T + \dreplug{T, T_2}, \dreplug{})$.
\item if $e_T = \dreplug{}$, explore each of the $\OhOf{n^2}$ trees that can be obtained from $T$ by one SPR operation and insert each such tree $t \notin V$ into $P$ with priority $(d_T + 1, 1, \one())$ and into $V$.
However, if any of $T$'s SPR neighbors are $T_2$ then we terminate the program and return the SPR distance of $d_T + 1$.
\end{enumerate}

\begin{restatable}{theorem}{thmsprdistance}
\label{thm:spr_distance}
The SPR distance between two unrooted trees $T_1$ and $T_2$ can be computed in $\OhOf{Y \, 16^d d^{1.5}}$ time, where $Y$ is the number of trees explored by the heuristic, and $d=\dspr{T_1, T_2}$.
Note that $Y=\OhOf{(28d)!!}$ after reducing $T_1$ and $T_2$.
\end{restatable}

\section{Experimental Evaluation}

We implemented our algorithms in the C++ program \texttt{uspr}~\cite{uspr}.
Given two unrooted trees, this software can compute their TBR $3$-approximation-based lower bound, TBR distance, replug distance, or SPR distance.
This program was tested on the prokaryote dataset of~\cite{beiko2005highways} which compares a phylogenetic supertree constructed by the Matrix Representation with Parsimony method to 22,437 individual gene trees ranging from 4--144 taxa.
This dataset has been widely used to test methods for computing rooted SPR distances~\cite{beiko2006phylogenetic,whidden2009unifying,whidden2010fast}, as well as the sole previous software for computing uSPR distances \texttt{sprdist} by Hickey et al.~\cite{hickey2008spr}.
(Note that this software is not the identically named software by different authors for computing rooted SPR distances~\cite{wu2009practical,wu2010fast}.)
In the remainder of this paper we refer only to the software of Hickey et al. as \texttt{sprdist}.

We computed SPR distances using \texttt{uspr} and \texttt{sprdist} on a computer cluster running Ubuntu 14.04 with the SLURM cluster management software.
We allocated one Intel Xeon X5650 CPU per computation and terminated instances which required more than 4096 MB of memory or 5 hours.
We tested the 5689 gene trees with 10 or more taxa.
We used the standard practice of comparing each gene tree to the subset of the supertree with identical taxa only.

\subsection{Running time}

\begin{figure}[t]
\includegraphics[width=\columnwidth]{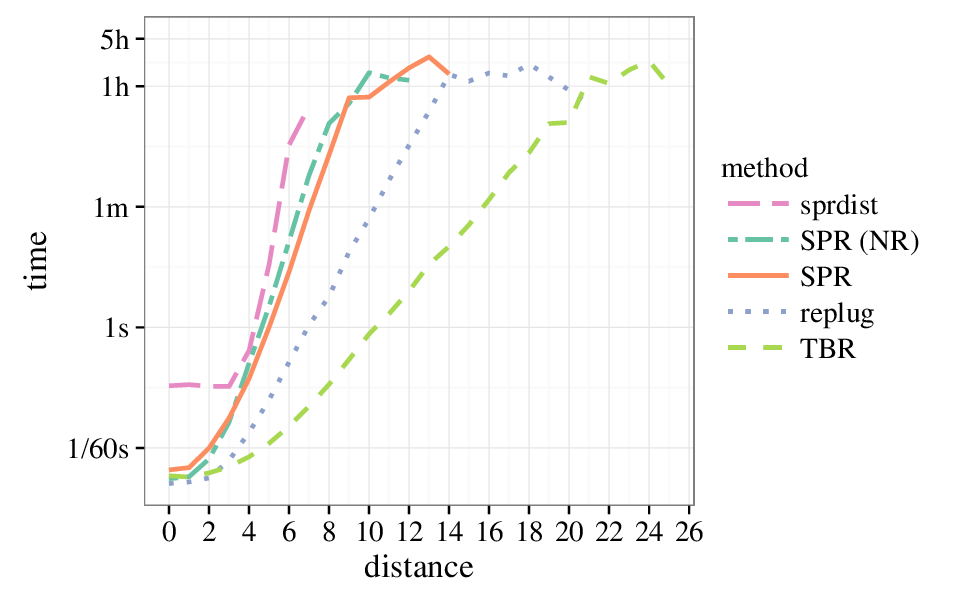}

	\caption{The mean time on a log-60 scale required as a function of distance to compute SPR distances using \texttt{sprdist} and SPR, replug, and TBR distances using \texttt{uspr}. The SPR test used all 4 bounding heuristics. The SPR (NR) test used all the bounding heuristics except the replug distance. Means were computed only over computations that succeeded for a given run within the time and memory limits.}
\label{fig:running_time}
\end{figure}

Our new algorithms allow us to compute much larger SPR distances than were previously possible, as well as compute the same distances with much less time and memory (Figure~\ref{fig:running_time}).
Note that Figure~\ref{fig:running_time} shows the mean time required by all computations that completed successfully given the time and memory limit, and summarizes different numbers of computations for different methods.
We computed SPR distances as large as 14 with \texttt{uspr}, double the maximum distance of 7 computable with \texttt{sprdist}.
\texttt{sprdist} also uses a graph exploration strategy, but without efficient heuristics to guide the search it typically reached the memory limit before the time limit was reached.
In contrast, \texttt{uspr} explores fewer trees but spends more time per tree and is therefore CPU bound and more scalable than \texttt{sprdist}.
Moreover, \texttt{uspr} found 176 instances with an SPR distance of 7 with a mean time of 53.29 seconds, nearly two orders of magnitude faster than the 33 instances with an SPR distance of 7 computed by \texttt{sprdist} with a mean time of 1808.96~seconds.

We tested our software with and without the replug distance heuristic to determine if the better distance estimate outweighed the extra computation required to compute the heuristic.
The replug distance heuristic was necessary to compute SPR distances of 13 or 14.
Our software was an average of 5x faster with the replug heuristic, taking a mean time of 20.59s compared to 98.62s with only the TBR heuristic on problem instances which both methods completed given the running time limit.
The replug heuristic greatly reduced the number of trees examined, and therefore the memory required, by a factor of about 33 to a mean of 2921.4 with the heuristic compared to 102,498.8 without the heuristic when each approach completed.

As expected, our methods can be used to compute much larger replug distances (at most 21) and TBR distances (up to 25) given 5 hours.
TBR distance computations are orders of magnitude faster than replug distance computations.
As both distances are used repeatedly during our SPR distance progressive A* search, the search time depends primarily on the time required to compute these distances as well as the number of such distances that must be computed.

\subsection{Completion}

\begin{figure}
\includegraphics[width=\columnwidth]{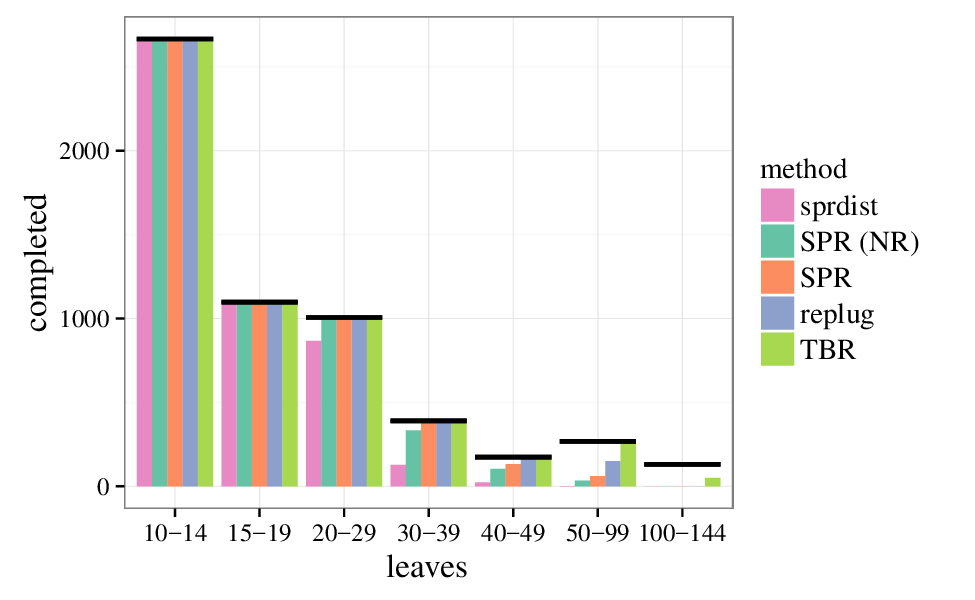}

\caption{The number of successfully computed SPR distances using \texttt{sprdist} and SPR, replug, and TBR distances using \texttt{uspr} given a time limit of 5 hours and memory limit of 4096 MB. Results are summarized by ranges of the number of leaves in the tree pairs, with the black lines indicating the total number of tree pairs in the given range.}
\label{fig:completion}
\end{figure}

Larger trees may have larger distances from a reference tree or supertree so we summarized the number of successfully computed distances for defined ranges of tree sizes (Figure~\ref{fig:completion}).
Our new SPR distance algorithm is practical for trees with up to 50 leaves.
We successfully computed distances for 132 of the 170 tree pairs with 40-49 leaves and 5142/5151 of the tree pairs with fewer than 40 leaves.
However, we were only able to compute SPR distances for 60 of the 261 tree pairs with 50-99 leaves.
This is in stark contrast to \texttt{sprdist} which can not reliably handle trees with more than 30 leaves, as it could only compute distances for 867 of the 1004 pairs of trees with 20-29 leaves, and 127 of the 386 tree pairs with 30-39 leaves.
We were able to reliably compute replug distances for trees with up to 65 leaves (114/133 successes in the 50-65 leaf range) and TBR distances for trees with up to 100 leaves (1 failure from the 261 tree pairs in the 50-99 leaf range).

\subsection{Mean distance}

\begin{figure}
\includegraphics[width=\columnwidth]{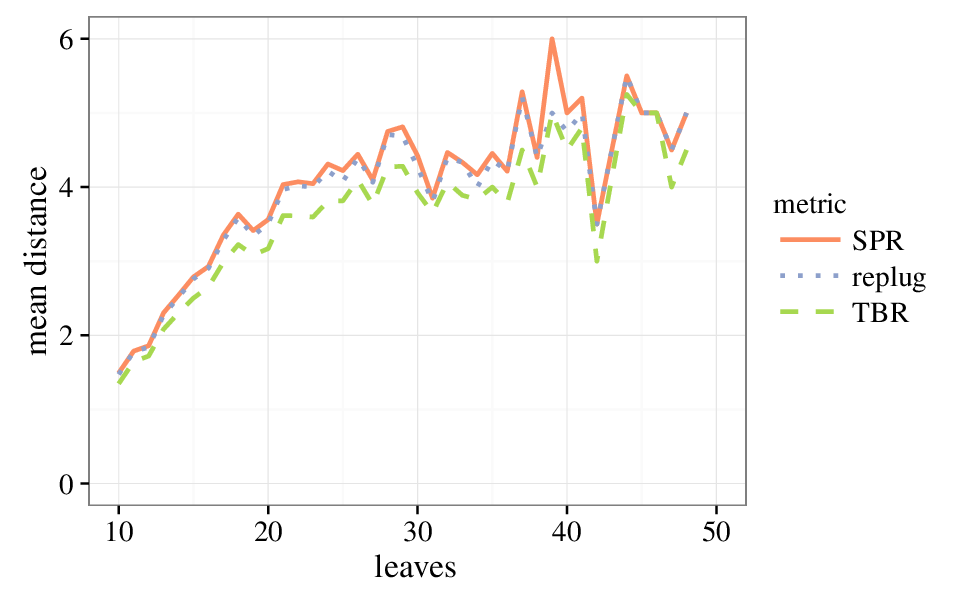}

\caption{Mean SPR, replug, and TBR distances of all tree pairs for which the SPR distance was succesfully computed using \texttt{uspr} given a time limit of 5 hours and memory limit of 4096 MB.}
\label{fig:mean_distance}
\end{figure}

Finally, we compared the mean SPR, replug, and TBR distance values for trees with a given number of leaves (Figure~\ref{fig:mean_distance}).
We computed these values only for the tree pairs for which we successfully computed the SPR distance to obtain a fair comparison.
We found that the replug distance is an excellent lower bound that closely tracks the SPR distance with a mean difference of 0.047 over 5335 tree pairs.
The TBR distance had a mean difference of 0.325 from the SPR distance and 0.278 from the replug distance.
In each of our completed tests, the TBR distance was at most 2 less than the replug distance and at most 3 less than the SPR distance, while the SPR distance never exceeded the replug distance by more than 1.

\section{Conclusions}

We have worked to extend understanding of and methods to calculate the SPR distance between unrooted trees in several directions.
We identified four properties of optimal SPR paths between unrooted trees which are atypical of NP-hard tree distance metrics that can be efficiently solved with an MAF formulation: both endpoints of an edge may be moved, the same endpoint of an edge may be moved multiple times, common clusters are not necessarily maintained, and common paths may be broken.
These observations suggest that an MAF formulation is not applicable for computing the SPR distance between unrooted trees.

To obtain an efficient search strategy we instead introduced a new lower bound on the SPR distance that we call the replug distance.
Although the computational complexity of this distance is unknown, our work shows that it is fixed parameter tractable and we conjecture that it is NP-hard to compute.
The replug distance captures important properties of the SPR distance while relaxing the requirement that intermediate structures be connected trees.
Moreover, we showed that the replug distance can be modeled using an MAF variant that we call a maximum endpoint agreement forest or MEAF.
We developed a two-phase fixed-parameter bounded search tree algorithm for the replug distance that runs in time $\OhOf{16^k(k^{1.5} + n)}$-time, where $k$ is the SPR distance between the trees and $n$ their number of leaves.
This algorithm works by exploring the set of maximal agreement forests of the trees.
Experiments suggest that these sets are typically small in practice~\cite{voorkamp2014maximal}.
Each such forest is then refined to a maximal endpoint agreement forest by solving boolean monotone $\cnfletwo$ sets of constraints.
These formulas naturally arise by considering sets of structures we call dead components in the trees given an mAF, and can be solved in polynomial time when converted to the minimimum edge cover problem of an appropriate constraint~graph.

Finally, we developed a new incremental heuristic search algorithm that we call progressive A* search.
This algorithm expands the search for a given tree outward from the initial tree by applying increasingly expensive but more accurate lower bound estimators.
The algorithm is applicable to any search problem that admits such a set of estimators, each of which is a lower bound on the next.
Progressive A* search computes the SPR distance $d$ between two unrooted trees in $\OhOf{(28d)!!16^dd^{1.5}}$-time.
Our implementation in the $\texttt{uspr}$ software package uses a TBR approximation, TBR distance, and replug distance as lower bounds.
Our results show that $\texttt{uspr}$ is nearly two orders of magnitude faster than the previous best software for computing SPR distances between unrooted trees.
In particular, our methods double the maximum SPR distance that can be computed given 5 hours from 7 to 14, and increase the size of trees that can be reliably compared from 30 to 50 leaves.
Moreover, our implementation of the replug and TBR distance metrics can handle distances as large as 21 and 25, respectively.
The replug and TBR algorithms were able to reliably handle trees with up to 65 and 100 leaves, respectively.

The development of initial fixed-parameter bounded search tree algorithms for the SPR distance between rooted trees quickly led to the current state of the art algorithms which can handle distances of 100 or more on trees with hundreds of leaves in only fractions of a second.
These improvements came from a combination of structural insights leading to improved branching rules and new reductions such as the cluster reduction rule which splits the compared trees into independently comparable subtrees.
Although unrooted trees are not clusterable with respect to the unrooted SPR distance, we conjecture that they are clusterable with respect to the replug distance.
As faster algorithms for the TBR and replug distance will immediately reduce the time required by our progressive A* search framework, we believe that refining these algorithms represents the best strategy for futher reducing the time required to compute SPR distances between unrooted trees.
In addition, techniques from the incremental heuristic search literature may also lead to improved algorithms for computing SPR distances.
Alternatively, further structural analysis may lead to a direct refinement procedure from MEAFs to unrooted SPR paths, representing another branch of study for future work.
Finally, it remains to extend our methods to nonbinary trees or to comparing sets of more than two trees, which are active avenues of research with respect to rooted trees.

\section*{Acknowledgments}

This work was funded by National Science Foundation award 1223057 and 1564137.
C. Whidden is a Simons Foundation Fellow of the Life Sciences Research Foundation.
The research of F. Matsen was supported in part by a Faculty Scholar grant from the Howard Hughes Medical Institute and the Simons Foundation.

\bibliographystyle{IEEEtran}
\bibliography{uspr}

\begin{thebibliography}{10}
\providecommand{\url}[1]{#1}
\csname url@samestyle\endcsname
\providecommand{\newblock}{\relax}
\providecommand{\bibinfo}[2]{#2}
\providecommand{\BIBentrySTDinterwordspacing}{\spaceskip=0pt\relax}
\providecommand{\BIBentryALTinterwordstretchfactor}{4}
\providecommand{\BIBentryALTinterwordspacing}{\spaceskip=\fontdimen2\font plus
\BIBentryALTinterwordstretchfactor\fontdimen3\font minus
  \fontdimen4\font\relax}
\providecommand{\BIBforeignlanguage}[2]{{%
\expandafter\ifx\csname l@#1\endcsname\relax
\typeout{** WARNING: IEEEtran.bst: No hyphenation pattern has been}%
\typeout{** loaded for the language `#1'. Using the pattern for}%
\typeout{** the default language instead.}%
\else
\language=\csname l@#1\endcsname
\fi
#2}}
\providecommand{\BIBdecl}{\relax}
\BIBdecl

\bibitem{hillis96}
D.~M. Hillis, C.~Moritz, and B.~K. Mable, Eds., \emph{Molecular
  Systematics}.\hskip 1em plus 0.5em minus 0.4em\relax Sinauer Associates,
  1996.

\bibitem{koonin2015turbulent}
E.~V. Koonin, ``The turbulent network dynamics of microbial evolution and the
  statistical tree of life,'' \emph{Journal of molecular evolution}, pp. 1--7,
  2015.

\bibitem{castro2012evolution}
E.~Castro-Nallar, M.~P{\'e}rez-Losada, G.~F. Burton, and K.~A. Crandall, ``The
  evolution of {HIV}: inferences using phylogenetics,'' \emph{Molecular
  phylogenetics and evolution}, vol.~62, no.~2, pp. 777--792, 2012.

\bibitem{haynes2012b}
B.~F. Haynes, G.~Kelsoe, S.~C. Harrison, and T.~B. Kepler, ``B-cell-lineage
  immunogen design in vaccine development with {HIV-1} as a case study,''
  \emph{Nature biotechnology}, vol.~30, no.~5, pp. 423--433, 2012.

\bibitem{helmus2007phylogenetic}
M.~R. Helmus, T.~J. Bland, C.~K. Williams, and A.~R. Ives, ``Phylogenetic
  measures of biodiversity,'' \emph{The American Naturalist}, vol. 169, no.~3,
  pp. E68--E83, 2007.

\bibitem{matsen2015phylogenetics}
F.~A. Matsen, ``Phylogenetics and the human microbiome,'' \emph{Systematic
  Biology}, vol.~64, no.~1, pp. e26--e41, 2015.

\bibitem{galtier2008dealing}
N.~Galtier and V.~Daubin, ``Dealing with incongruence in phylogenomic
  analyses,'' \emph{Philosophical Transactions of the Royal Society B:
  Biological Sciences}, vol. 363, no. 1512, pp. 4023--4029, 2008.

\bibitem{beiko2005highways}
R.~G. Beiko, T.~J. Harlow, and M.~A. Ragan, ``Highways of gene sharing in
  prokaryotes,'' \emph{Proceedings of the National Academy of Sciences of the
  United States of America}, vol. 102, no.~40, pp. 14\,332--14\,337, 2005.

\bibitem{whidden2014supertrees}
C.~Whidden, N.~Zeh, and R.~G. Beiko, ``Supertrees based on the subtree
  prune-and-regraft distance,'' \emph{Systematic Biology}, vol.~63, no.~4, pp.
  566--581, 2014.

\bibitem{pisani2007supertrees}
D.~Pisani, J.~A. Cotton, and J.~O. McInerney, ``Supertrees disentangle the
  chimerical origin of eukaryotic genomes,'' \emph{Molecular Biology and
  Evolution}, vol.~24, no.~8, pp. 1752--1760, 2007.

\bibitem{Steel2008-pn}
\BIBentryALTinterwordspacing
M.~Steel and A.~Rodrigo, ``Maximum likelihood supertrees,'' \emph{Systematic
  Biology}, vol.~57, no.~2, pp. 243--250, Apr. 2008. [Online]. Available:
  \url{http://dx.doi.org/10.1080/10635150802033014}
\BIBentrySTDinterwordspacing

\bibitem{bansal2010robinson}
M.~S. Bansal, J.~G. Burleigh, O.~Eulenstein, and D.~Fern{\'a}ndez-Baca,
  ``Robinson-{F}oulds supertrees.'' \emph{Algorithms for Molecular Biology},
  vol.~5, no.~18, pp. 1--12, 2010.

\bibitem{robinson81}
D.~F. Robinson and L.~R. Foulds, ``Comparison of phylogenetic trees,''
  \emph{Mathematical Biosciences}, vol.~53, no. 1-2, pp. 131--147, 1981.

\bibitem{day85}
W.~H.~E. Day, ``Optimal algorithms for comparing trees with labeled leaves,''
  \emph{Journal of Classification}, vol.~2, no.~1, pp. 7--28, 1985.

\bibitem{brodal2004computing}
G.~S. Brodal, R.~Fagerberg, and C.~N. Pedersen, ``Computing the quartet
  distance between evolutionary trees in time {O} (n log n),''
  \emph{Algorithmica}, vol.~38, no.~2, pp. 377--395, 2004.

\bibitem{owen2011fast}
M.~Owen and J.~S. Provan, ``A fast algorithm for computing geodesic distances
  in tree space,'' \emph{IEEE/ACM Transactions on Computational Biology and
  Bioinformatics (TCBB)}, vol.~8, no.~1, pp. 2--13, 2011.

\bibitem{baroni05}
M.~Baroni, S.~Gr{\"u}newald, V.~Moulton, and C.~Semple, ``Bounding the number
  of hybridisation events for a consistent evolutionary history,''
  \emph{Journal of Mathematical Biology}, vol.~51, no.~2, pp. 171--182, 2005.

\bibitem{Bruen2008-wx}
\BIBentryALTinterwordspacing
T.~C. Bruen and D.~Bryant, ``Parsimony via consensus,'' \emph{Systematic
  Biology}, vol.~57, no.~2, pp. 251--256, 1~Apr. 2008. [Online]. Available:
  \url{http://sysbio.oxfordjournals.org/content/57/2/251.abstract}
\BIBentrySTDinterwordspacing

\bibitem{kelk2014complexity}
S.~Kelk and M.~Fischer, ``On the complexity of computing {MP} distance between
  binary phylogenetic trees,'' \emph{arXiv preprint arXiv:1412.4076}, 2014.

\bibitem{moulton2015parsimony}
V.~Moulton and T.~Wu, ``A parsimony-based metric for phylogenetic trees,''
  \emph{Advances in Applied Mathematics}, vol.~66, pp. 22--45, 2015.

\bibitem{beiko2006phylogenetic}
R.~G. Beiko and N.~Hamilton, ``Phylogenetic identification of lateral genetic
  transfer events,'' \emph{BMC Evolutionary Biology}, vol.~6, no.~1, p.~15,
  2006.

\bibitem{maddison97}
W.~P. Maddison, ``Gene trees in species trees,'' \emph{Systematic Biology},
  vol.~46, no.~3, pp. 523--536, 1997.

\bibitem{nakhleh05}
L.~Nakhleh, T.~Warnow, C.~R. Lindner, and K.~St.~John, ``Reconstructing
  reticulate evolution in species---theory and practice,'' \emph{Journal of
  Computational Biology}, vol.~12, no.~6, pp. 796--811, 2005.

\bibitem{Price2010-fi}
M.~N. Price, P.~S. Dehal, and A.~P. Arkin, ``{FastTree} 2--approximately
  maximum-likelihood trees for large alignments,'' \emph{PLoS One}, vol.~5,
  no.~3, p. e9490, 2010.

\bibitem{Stamatakis2006-yz}
A.~Stamatakis, ``{RAxML-VI-HPC}: maximum likelihood-based phylogenetic analyses
  with thousands of taxa and mixed models,'' \emph{Bioinformatics}, vol.~22,
  no.~21, pp. 2688--2690, 23~Aug. 2006.

\bibitem{Ronquist2012-hi}
\BIBentryALTinterwordspacing
F.~Ronquist, M.~Teslenko, P.~van~der Mark, D.~L. Ayres, A.~Darling,
  S.~H{\"{o}}hna, B.~Larget, L.~Liu, M.~A. Suchard, and J.~P. Huelsenbeck,
  ``{MrBayes} 3.2: efficient {Bayesian} phylogenetic inference and model choice
  across a large model space,'' \emph{Systematic Biology}, vol.~61, no.~3, pp.
  539--542, 22~Feb. 2012. [Online]. Available:
  \url{http://dx.doi.org/10.1093/sysbio/sys029}
\BIBentrySTDinterwordspacing

\bibitem{bouckaert2014beast}
R.~Bouckaert, J.~Heled, D.~K{\"u}hnert, T.~Vaughan, C.-H. Wu, D.~Xie, M.~A.
  Suchard, A.~Rambaut, and A.~J. Drummond, ``{BEAST} 2: a software platform for
  {Bayesian} evolutionary analysis,'' \emph{PLoS Computational Biology},
  vol.~10, no.~4, p. e1003537, 2014.

\bibitem{atkins2015extremal}
\BIBentryALTinterwordspacing
R.~Atkins and C.~McDiarmid, ``Extremal distances for subtree transfer
  operations in binary trees,'' \emph{arXiv preprint arXiv:1509.00669}, 2015.
  [Online]. Available: \url{http://arxiv.org/abs/1509.00669}
\BIBentrySTDinterwordspacing

\bibitem{Ding2011-bj}
Y.~Ding, S.~Gr{\"{u}}newald, and P.~J. Humphries, ``On agreement forests,''
  \emph{Journal of Combinatorial Theory, Series A}, vol. 118, no.~7, pp.
  2059--2065, Oct. 2011.

\bibitem{lakner2008efficiency}
C.~Lakner, P.~Van Der~Mark, J.~P. Huelsenbeck, B.~Larget, and F.~Ronquist,
  ``Efficiency of {Markov chain Monte Carlo} tree proposals in {Bayesian}
  phylogenetics,'' \emph{Systematic Biology}, vol.~57, no.~1, pp. 86--103,
  2008.

\bibitem{hohna2012guided}
S.~H{\"o}hna and A.~J. Drummond, ``Guided tree topology proposals for
  {Bayesian} phylogenetic inference,'' \emph{Systematic Biology}, vol.~61,
  no.~1, pp. 1--11, 2012.

\bibitem{whidden2015quantifying}
\BIBentryALTinterwordspacing
C.~Whidden and F.~A. Matsen~IV, ``Quantifying {MCMC} exploration of
  phylogenetic tree space,'' \emph{Systematic Biology}, vol.~64, no.~3, pp.
  472--491, 27~Jan. 2015. [Online]. Available:
  \url{http://dx.doi.org/10.1093/sysbio/syv006}
\BIBentrySTDinterwordspacing

\bibitem{bordewich07}
M.~Bordewich and C.~Semple, ``Computing the minimum number of hybridization
  events for a consistent evolutionary history,'' \emph{Discrete Applied
  Mathematics}, vol. 155, no.~8, pp. 914--928, 2007.

\bibitem{bordewich2005computational}
------, ``On the computational complexity of the rooted subtree prune and
  regraft distance,'' \emph{Annals of Combinatorics}, vol.~8, no.~4, pp.
  409--423, 2005.

\bibitem{hickey2008spr}
G.~Hickey, F.~Dehne, A.~Rau-Chaplin, and C.~Blouin, ``{SPR} distance
  computation for unrooted trees,'' \emph{Evolutionary Bioinformatics Online},
  vol.~4, p.~17, 2008.

\bibitem{whidden2010fast}
C.~Whidden, R.~G. Beiko, and N.~Zeh, ``Fast {FPT} algorithms for computing
  rooted agreement forests: theory and experiments,'' in \emph{Experimental
  algorithms}.\hskip 1em plus 0.5em minus 0.4em\relax Springer, 2010, pp.
  141--153.

\bibitem{dudas2014reassortment}
G.~Dudas, T.~Bedford, S.~Lycett, and A.~Rambaut, ``Reassortment between
  {Influenza B} lineages and the emergence of a co-adapted {PB1-PB2-HA} gene
  complex,'' \emph{Molecular biology and evolution}, p. msu287, 2014.

\bibitem{whidden2015ricci}
C.~Whidden and F.~A.~M. IV, ``Ricci-{O}llivier curvature of the rooted
  phylogenetic subtree-prune-regraft graph,'' \emph{proceedings of the
  Thirteenth Workshop on Analytic Algorithmics and Combinatorics (ANALCO16)},
  2016.

\bibitem{hein96}
J.~Hein, T.~Jiang, L.~Wang, and K.~Zhang, ``On the complexity of comparing
  evolutionary trees,'' \emph{Discrete Applied Mathematics}, vol.~71, no. 1-3,
  pp. 153--169, 1996.

\bibitem{allen01}
B.~L. Allen and M.~Steel, ``Subtree transfer operations and their induced
  metrics on evolutionary trees,'' \emph{Ann. Comb.}, vol.~5, no.~1, pp. 1--15,
  2001.

\bibitem{wu2009practical}
Y.~Wu, ``A practical method for exact computation of subtree prune and regraft
  distance,'' \emph{Bioinformatics}, vol.~25, no.~2, pp. 190--196, 2009.

\bibitem{bonet2009efficiently}
M.~L. Bonet and K.~St.~John, ``Efficiently calculating evolutionary tree
  measures using {SAT},'' in \emph{Proceedings of the 12th International
  Conference on Theory and Applications of Satisfiability Testing}, ser.
  Lecture Notes in Computer Science, vol. 5584.\hskip 1em plus 0.5em minus
  0.4em\relax Springer-Verlag, 2009, pp. 4--17.

\bibitem{whidden2013fixed}
C.~Whidden, R.~G. Beiko, and N.~Zeh, ``Fixed-{P}arameter algorithms for maximum
  agreement forests,'' \emph{SIAM J. Comput.}, vol.~42, no.~4, pp. 1431--1466,
  2013.

\bibitem{shi2014improved}
F.~Shi, J.~You, and Q.~Feng, ``Improved approximation algorithm for maximum
  agreement forest of two trees,'' in \emph{Frontiers in Algorithmics}.\hskip
  1em plus 0.5em minus 0.4em\relax Springer, 2014, pp. 205--215.

\bibitem{chen2013faster}
Z.-Z. Chen, Y.~Fan, and L.~Wang, ``Faster exact computation of {rSPR}
  distance,'' \emph{Journal of Combinatorial Optimization}, vol.~29, no.~3, pp.
  605--635, 2013.

\bibitem{whidden2015multifurcating}
C.~Whidden, R.~G. Beiko, and N.~Zeh, ``Fixed-{Parameter} and approximation
  algorithms for maximum agreement forests of multifurcating trees,''
  \emph{Algorithmica}, pp. 1--36, 2015, doi: 10.1007/s00453-015-9983-z.

\bibitem{chen2015parameterized}
J.~Chen, J.-H. Fan, and S.-H. Sze, ``Parameterized and approximation algorithms
  for maximum agreement forest in multifurcating trees,'' \emph{Theoretical
  Computer Science}, vol. 562, pp. 496--512, 2015.

\bibitem{shi2014approximation}
F.~Shi, J.~Chen, Q.~Feng, and J.~Wang, ``Approximation algorithms for maximum
  agreement forest on multiple trees,'' in \emph{Computing and Combinatorics:
  20th International Conference, COCOON 2014, Atlanta, GA, USA, August 4-6,
  2014, Proceedings}, vol. 8591.\hskip 1em plus 0.5em minus 0.4em\relax
  Springer, 2014, p. 381.

\bibitem{linz2011cluster}
S.~Linz and C.~Semple, ``A cluster reduction for computing the subtree distance
  between phylogenies,'' \emph{Annals of Combinatorics}, vol.~15, no.~3, pp.
  465--484, 2011.

\bibitem{whidden2016chain}
C.~Whidden and F.~A. Matsen~IV, ``Chain reduction preserves the unrooted
  subtree prune-and-regraft distance,'' \emph{arXiv preprint arXiv:1611.02351},
  2016.

\bibitem{bonet2010complexity}
M.~L. Bonet and K.~St~John, ``On the complexity of {uSPR} distance,''
  \emph{IEEE/ACM Transactions on Computational Biology and Bioinformatics
  (TCBB)}, vol.~7, no.~3, pp. 572--576, 2010.

\bibitem{shi2013parameterized}
F.~Shi, J.~Chen, Q.~Feng, and J.~Wang, ``Paramaterized algorithms for maximum
  agreement forest on multiple trees,'' in \emph{Computing and Combinatorics:
  19th International Conference, COCOON 2013, Hangzhou, China, June 21-23,
  2013. Proceedings}, vol. 8591.\hskip 1em plus 0.5em minus 0.4em\relax
  Springer Berlin Heidelberg, 2013, pp. 567--578.

\bibitem{baroni2006hybrids}
M.~Baroni, C.~Semple, and M.~Steel, ``Hybrids in real time,'' \emph{Systematic
  Biology}, vol.~55, no.~1, pp. 46--56, 2006.

\bibitem{whidden2009unifying}
C.~Whidden and N.~Zeh, ``A unifying view on approximation and {FPT} of
  agreement forests,'' in \emph{Proceedings of the 9th International Workshop,
  WABI 2009}, ser. Lecture Notes in Bioinformatics, vol. 5724.\hskip 1em plus
  0.5em minus 0.4em\relax Springer-Verlag, 2009, pp. 390--401.

\bibitem{chen2013parameterized}
J.~Chen, J.-H. Fan, and S.-H. Sze, ``Parameterized and approximation algorithms
  for the {MAF} problem in multifurcating trees,'' in \emph{Graph-Theoretic
  Concepts in Computer Science}.\hskip 1em plus 0.5em minus 0.4em\relax
  Springer, 2013, pp. 152--164.

\bibitem{voorkamp2014maximal}
J.~Voorkamp, ``Maximal acyclic agreement forests,'' \emph{Journal of
  Computational Biology}, vol.~21, no.~10, pp. 723--731, 2014.

\bibitem{porschen2007algorithms}
S.~Porschen and E.~Speckenmeyer, ``Algorithms for variable-weighted {2-SAT} and
  dual problems,'' in \emph{Theory and Applications of Satisfiability
  Testing--SAT 2007}.\hskip 1em plus 0.5em minus 0.4em\relax Springer, 2007,
  pp. 173--186.

\bibitem{hart1968formal}
P.~E. Hart, N.~J. Nilsson, and B.~Raphael, ``A formal basis for the heuristic
  determination of minimum cost paths,'' \emph{Systems Science and Cybernetics,
  IEEE Transactions on}, vol.~4, no.~2, pp. 100--107, 1968.

\bibitem{uspr}
C.~Whidden, ``uspr,'' \url{https://github.com/cwhidden/uspr}, 2015.

\bibitem{wu2010fast}
Y.~Wu and J.~Wang, ``Fast computation of the exact hybridization number of two
  phylogenetic trees,'' in \emph{Bioinformatics Research and
  Applications}.\hskip 1em plus 0.5em minus 0.4em\relax Springer, 2010, pp.
  203--214.

\bibitem{garey2002computers}
M.~R. Garey and D.~S. Johnson, \emph{Computers and intractability}.\hskip 1em
  plus 0.5em minus 0.4em\relax {W. H.} Freeman, 1979.

\bibitem{micali1980v}
S.~Micali and V.~V. Vazirani, ``An ${O}(\sqrt{V} {E})$ algorithm for finding
  maximum matching in general graphs,'' in \emph{Foundations of Computer
  Science, 1980., 21st Annual Symposium on}.\hskip 1em plus 0.5em minus
  0.4em\relax IEEE Computer Society Press, 1980, pp. 17--27.

\end{thebibliography}

%
%

\clearpage

\beginsupplement

\appendices

\section{Proofs and Supplemental Figures}

\subsection{Additional definitions}

We first specify some definitions that are only used in the following proofs.

A set of four leaves $\{a,b,c,d\}$ of an unrooted tree $T$ form a \emph{quartet} $\quartet{ab}{cd}$ when the path from $a$ to $b$ and the path from $c$ to $d$ are vertex-disjoint.
A forest contains each of the quartets of its individual component trees.
Given a tree $T$ and a forest $F$, we say a quartet $\quartet{ab}{cd}$ of $T$ is \emph{incompatible} with $F$ if its leaves do not all belong to the same component of $F$ or define a different quartet in $F$ (e.g. $\quartet{ac}{bd}$).
Whidden and Zeh~\cite{whidden2009unifying} observed that:

\begin{obs}
\label{obs:incompatible}
Let $T_1$ and $T_2$ be unrooted $X$-trees, $F$ a forest of $T_2$ and $E$ a set of edges of $F$ such that $F - E$ yields an agreement forest of $T_1$ and $T_2$.
If $\quartet{ab}{cd}$ is a quartet of $T_1$ incompatible with $F$, then either $a \noreach{F-E} b$, $a \noreach{F-E} c$, or $c \noreach{F-E} d$.
\end{obs}

\subsection{Proofs}

\lemreplugniceproperties*
\begin{proof}
	Consider an SAF $F$ of $T_1$ and $T_2$ that permits $M$.
	Assume, for the purpose of obtaining a contradiction, the opposite of the first claim: that two distinct moves $m_i$ and $m_j$ move the same endpoint $u$ of a connection $c$ of $F$.
	We choose $i$ and $j$ such that no move $m_k$ moves endpoint $u$ of $c$, for all $i < k < j$.
	Consider the sequence of replug moves $M' = m_1, m_2, \ldots, m_{i-1}, m_{i+1}, m_{i+2}, \ldots, m_k$.
	By Observation~\ref{obs:replug-independence} and the fact that no moves of $M$ between $m_i$ and $m_j$ move endpoint $u$, this is a valid sequence of replug moves at least through move $m_{j-1}$.
	Moreover, $m_j$ is then a valid move that results in the same forest as applying $m_j$ in sequence $M$.
	Therefore applying the remainder of the sequence results in the same tree, which implies that $M$ and $M'$ both result in tree $T_2$.
	The fact that $M'$ is a smaller sequence contradicts the optimality of $M$, proving the claim.

	The second and third claims follow similarly, by substituting moves that break and reform a common cluster and common path for $m_i$ and $m_j$ in the above argument.
\end{proof}

\begin{figure}
\includegraphics[width=\columnwidth]{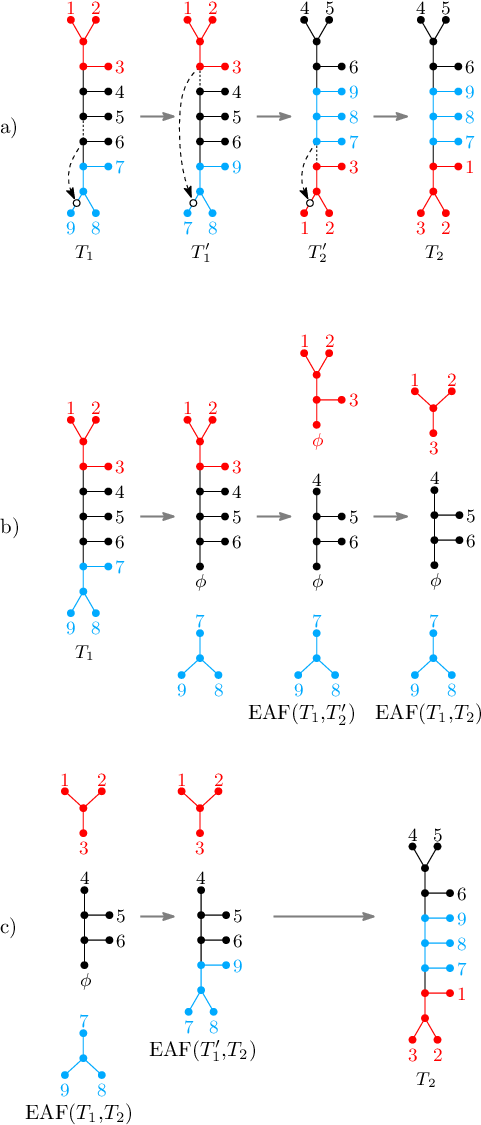}

	\caption{The conversion processes from Theorem~\ref{thm:meaf}. (a) The trees from Figure~\ref{fig:three-spr}.
	(b) An EAF is built iteratively from each replug move by cutting the changed edge and attaching a $\phi$ node to the fixed endpoint (tail of the dashed arrow from (a)).
	If the other endpoint is a $\phi$ node then both sides of the edge have moved so the $\phi$ node is removed and none added.
	(c) Given just the EAF, a start tree and an end tree, a minimal replug path can be found by applying replug moves that merge components of the EAF. If there is a $\phi$ node we apply a replug move and merge two components. When no $\phi$ nodes remain, we merge components by applying TBR moves (i.e. two replug moves).}
\label{fig:meaf}
\end{figure}

\thmmeaf*
\begin{proof}
We prove this by induction.
We first show that $\dreplug{T_1, T_2} \ge \weight(T_1, T_2)$ by induction on $\dreplug{T_1, T_2}$.
	To do so, given an optimal sequence of replug moves transforming $T_1$ into $T_2$ (e.g. Fig~\ref{fig:meaf}(a)), we iteratively construct an EAF $F$ of $T_1$ and $T_2$ such that $\weight(F) = \dreplug{T_1, T_2}$ (Fig~\ref{fig:meaf}(b)).
Because cutting an endpoint edge set strictly increases the weight of a forest, $\weight(T_1, T_2) = 0$ if and only if $T_1 = T_2$ (i.e. $\dreplug{T_1, T_2} = 0$).
This forms the base case of the induction.

Suppose that our claim holds for every pair of trees with a replug distance less than $d$ and that $\dreplug{T_1, T_2} = d$.
Moreover, let $m_1, m_2, \ldots, m_d$ be a sequence of replug moves that transform $T_1$ into $T_2$.
If $d=1$ then let $T'_2 = T_1$.
	Otherwise, let $T'_2$ be the result of applying the sequence of moves $m_1, m_d, \ldots, m_{d-1}$ to $T_1$, so that $m_d$ transforms $T'_2$ into $T_2$.
	Thus, by the inductive hypothesis, there exists an EAF $F'$ of $T_1$ and $T'_2$ such that $\weight(F') = d-1$. For example, $F'$ is labeled EAF($T_1$, $T_2'$) in Fig~\ref{fig:meaf}(b).
Let $e = (u,v) $ be the edge moved by $m_d$, where $u$ is the endpoint of $e$ that is not moved and $v$ is the endpoint of $e$ that is moved.
Now, we wish to cut edge $e$ in $F'$ to obtain an EAF of $T_1$ and $T_2$.
There are two cases depending on whether $e$ has already been cut in $F'$.

First, assume that $e$ has been cut.
An optimal sequence of replug moves never moves the same endpoint $v$ of an edge $e$ twice by Lemma~\ref{lem:replug-nice-properties}.
	This implies that the replug move in $m_1, m_2, m_3, \ldots, m_{d-1}$ which corresponds to cutting edge $e$ must have moved endpoint $u$ rather than endpoint $v$.
Thus, $F'$ contains an edge $e' = (v,\phi)$.
Moreover, $T'_2$ and $T_2$ differ only in the location of endpoint $v$ of edge $e$.
We remove the $\phi$ node from $e'$ to obtain the forest $F'' = F' \div \set{\eedge{e'}{\emptyset}}$.
	$F''$ is an EAF of $T_1$ and $T_2$ with weight $d$, which proves the claim. For example,  $F''$ is labeled EAF($T_1$,$T_2$) in Fig~\ref{fig:meaf}(b).

If $e$ has not been cut, then take $F = F' \div \set{\eedge{e}{\set{u}}}$.
Because $T'_2$ and $T_2$ differ only in the location of endpoint $v$ of edge $e$, $F$ is an EAF of $T_1$ and $T_2$.
Moreover, $F$ has weight $d$, so the claim holds.
	For example, EAF($T_1$,$T_2'$) in Fig~\ref{fig:meaf}(b) can be obtained in this way from the previous EAF of $T_1$ and $T'_1$.

	For the other direction, we show that $\dreplug{T_1,T_2} \le \weight(T_1, T_2)$ by constructing a sequence of replug moves $m_1, m_2, \ldots, m_{\weight(T_1, T_2)}$ that transform $T_1$ into $T_2$ (e.g. Fig~\ref{fig:meaf}(a) using Fig~\ref{fig:meaf}(c)).
If $\weight(T_1, T_2) = 0$ then $T_1 = T_2$ and we construct an empty sequence of replug moves that transform $T_1$ into $T_2$, forming the base case.
By induction on $\weight(T_1, T_2)$, suppose that our claim holds for every pair of trees with an MEAF of weight less than $\weight$, and that $\weight(T_1, T_2) = \weight$.
	Let $F$ be an MEAF of $T_1$ and $T_2$ (labeled EAF($T_1$,$T_2$) in Fig~\ref{fig:meaf}(c)).
Let $E$ be an endpoint edge set such that $T_1 \div E = F$.
There are two cases: either $F$ is an MAF of $T_1$ and $T_2$ ($F$ has no $\phi$ nodes), or $F$ is not an MAF of $T_1$ and $T_2$ ($F$ has $\phi$ nodes).
	If $F$ is an MAF, then we can map it to a set of $\weight/2$ TBR operations, and therefore a set of $\weight$ replug operations and the claim holds. For example, EAF($T_1'$,$T_2$) in Fig~\ref{fig:meaf}(c) is an MAF of $T_1'$ and $T_2$ and we can transform $T_1'$ into $T_2$ by a single TBR operation. We can then separate that TBR operation into two equivalent replug operations.

So, assume that $F$ is not an MAF.
Then some component $C_i$ contains a $\phi$ node $x_1$.
By the definition of an agreement forest, there exists at least one component $C_j$ of $F$ that is ``effectively adjacent'' to $C_i$.
	That is, $C_i$ and $C_j$ are joined by a path $P = x_1, x_2, \ldots, x_q$ of nodes in $T_2$ such that $x_1 \in C_i$, $x_q \in C_j$, and $x_l \notin F$, for all $1 < l < q$.
Let $e_1$ be the edge adjacent to $x_1$ in $T_1$.
	Note that the $\phi$ node implies that $\eedge{e_1}{\set{x_1}} \in E$.
Let $y_q$ be the neighbor of $x_q$ in $C_j$ such that $x_{q-1}$ is on the path from $y_q$ to $x_q$ in $T_2$.
	We apply a replug move $m$ to connect $x_1$ to the edge adjacent to $x_q$ that is closest to $y_q$, resulting in $T_1'$.
	The node introduced by this replug move is $x_{q-1}$, which connects $C_i$ and $C_j$ in $T_1'$.
Thus, $T_1' \div (E \setminus \set{\eedge{e_1}{\set{x_1}}})$ is an MEAF of $T_1'$ and $T_2$ with weight $\weight-1$.
By the inductive hypothesis, we can construct a sequence of $\weight-1$ replug moves $M$ transforming $T_1'$ to $T_2$.
Therefore, the sequence of $\weight$ replug moves starting with $m$ and then applying the moves in $M$ transforms $T_1$ to $T_2$, and the claim holds.
	For example, in Fig~\ref{fig:meaf}(c) we see that leaves 6 and 9 are adjacent in $T_2$.
    We apply a replug move to $T_1$ that merges components of EAF($T_1$,$T_2$) to obtain $T'_1$ and EAF($T'_1$,$T_2$).
\end{proof}

\coreaftoreplug*
\begin{proof}
	We will show that we can find, in linear time, a replug move on $T_1$ that results in a tree $T'$ such that $\dreplug{T',T_2} = \weight(F) - 1$.
	Recursively applying this procedure $\weight(F)$ times proves the claim.
	The proof of Theorem~\ref{thm:meaf} provides a method to find such a move, so we show here that this method can be implemented to take linear time.
	We can construct mappings $\psi$ and $\psi^{-1}$ from the trees to the EAF in linear time by Lemma~\ref{lem:map-edges}.

	The steps of this procedure are to (1) select a component $C_i$ of $F$ with a $\phi$ node if any exist and an arbitrary component otherwise, (2) identify a second component $C_j$ of $F$ that is effectively adjacent to $C_i$ in $T_2$, and (3) apply the corresponding move to attach $C_i$ to $C_j$.
	By effectively adjacent, we again mean that we must choose $C_j$ such that $C_i$ and $C_j$ are joined by a path $P = x_1, x_2, \ldots, x_m$ of nodes in $T_2$ that does not include any nodes of another component $C_k$ or any nodes of $C_i$ other than $x_1$.
	We can find a component $C_i$ of $F$ with a $\phi$ node $x_1$, if any exist, in linear time by traversing the forest.
	Otherwise, we can find an arbitrary component $C_i$ in linear time.
	To find a valid component $C_j$ we first label the nodes of $F$ according to their component numbers, which takes linear time.
	For each node $n$ of $F$, we then label the corresponding $T_2$ node $\psi^{-1}(T_2, n)$ with $n$'s component number.
	Finally, we apply a traversal of $T_2$ that starts from $\psi^{-1}(T_2, x_1)$ and finds such a path by not visiting nodes labeled $i$ and terminating upon finding the first otherwise labeled node $x_m$.
	We can then apply a replug move that connects $x_i$ to $x_m$, resulting in a tree $T'$ with an EAF of $T'$ and $T_2$ with weight $\weight(F) - 1$.

	The steps are similar when no component has a $\phi$ node, except that we begin our traversal from an arbitrary node $x_1$ of the chosen $C_i$ that is adjacent to a node $x_2$ of $T_2$ such that $\psi(x_2) = \emptyset$.

\end{proof}

\thmreplugbound*
\begin{proof}
We showed in the proof of Theorem~\ref{thm:meaf} that the TBR distance is a lower bound for the replug distance.
To see that the replug distance is a lower bound for the SPR distance, it suffices to note that every SPR move is also a replug move, that is, every sequence of SPR moves is also a valid sequence of replug moves.
\end{proof}

\lemtbr*
\begin{proof}
We first observe that the set returned by $\alg{T_1, T_2, k}$ cannot contain an object that is not an agreement forest, as Step~\ref{case:success} will never apply.
Now, suppose that the algorithm misses a maximal agreement forest $F$ of $T_1$ and $T_2$ that can be obtained by cutting fewer than $k$ edges.
Let $E$ be an edge set such that $F = T_2 \div E$.
Steps~\ref{case:tbr:is}~and~\ref{case:tbr:mpe} are the only steps that modify $F_2$.
Thus, every path of recursive invocations terminates in a Step~\ref{case:abort} and contains an invocation $\alg{F_1, F_2, k'}$ that applied Step~\ref{case:tbr:is} or Step~\ref{case:tbr:mpe} such that every cut edge $e' \in F_2$ partitions a pair of leaves that are in the same component of $F$.

So, consider an arbitary path of recursive calls and let $\alg{F_1, F_2, k'}$ be the first call on this path that applied Step~\ref{case:tbr:is} or Step~\ref{case:tbr:mpe} such that every cut edge $e' \in F_2$ partitions a pair of leaves $l_1$ and $l_2$ that are in the same component of $F$.
In Step~\ref{case:tbr:is}, the fact that $a$ and $c$ are siblings in $F_1$ but $a \noreach{F_2} c$ implies that either $a \noreach{F} x$ for all $x \in (R_t \setminus \set{a})$, or $c \noreach{F} x$ for all $x \in (R_t \setminus \set{c})$, a contradiction.

Observation~\ref{obs:incompatible} implies the same for Step~\ref{case:tbr:mpe} with respect to at least one of $a$, $b_1$, $b_q$, or $c$.
We observe that Step~\ref{case:tbr:mpe} is the expanded form of the union of four recursive~calls:
\begin{align*}
	&\alg{F_1, F_2 \div \set{\edge{a}}, k-1}\ \cup \alg{F_1, F_2 \div \set{\edge{b_1}}, k-1}\ \\
	&\cup \alg{F_1, F_2 \div \set{\edge{b_q}}, k-1}\ \cup \alg{F_1, F_2 \div \set{\edge{c}}, k-1}.
\end{align*}
To see this, first note that both Step~\ref{case:tbr:mpe} and these four calls include calls cutting $\edge{a}$ and $\edge{c}$.
Thus, both sets of calls will find any mAF that is a forest of $F_2 \div \set{e_a}$ or $F_2 \div \set{e_c}$.
Now, observe that an application of the calls cutting $\edge{b_1}$ or $\edge{b_q}$ can not remove the incompatible sibling pair $\set{a,c}$ (unless $q=2$ and both sets of recursive calls are identical).
Expanding all possible combinations of repeated applications of these calls results in the cases of Step~\ref{case:tbr:mpe}.
Thus, one of the cut edge sets partitions only leaves which are partitioned in $F$, also a contradiction.
\end{proof}

\thmtbr*
\begin{proof}
This proof follows similar arguments to those of Whidden and Zeh~\cite{whidden2009unifying}.
$\alg{T_1, T_2, k}$ is a bounded search tree algorithm which proceeds to a depth at most $k$ and whose worst case behaviour is fully defined by the branching factor of the recurrence relation of Step~\ref{case:tbr:mpe} which is maximized when $q=2$ (i.e. 4 single edge cut invocations).
Thus, there are at most $4^k$ recursive invocations in total, each of which requires linear time (using the data structures from~\cite{whidden2009unifying}), excluding the cost of recursion and set union operations.
\end{proof}

\lemmapedges*
\begin{proof}
We adapt the procedure of Lemma 4.3 in \cite{whidden2013fixed} to the unrooted case.
This procedure constructs, in linear time, a cycle graph structure that is essentially the union of $F$, $E_1$, $E_2$, and an explicit mapping from nodes of $T_1$ and $T_2$ to $F$ and vice versa.
Although written in terms of rooted trees, the fact that the root is a labeled leaf implies that this procedure also applies to unrooted trees by choosing an arbitrary leaf as the ``root''.
We can obtain $E_1$ and $E_2$ from this structure by (1) iterating through the edges of the cycle graph and (2) applying the respective node mapping to identify the $T_1$ or $T_2$ edge.
\end{proof}

\begin{figure}
	\includegraphics[width=\columnwidth]{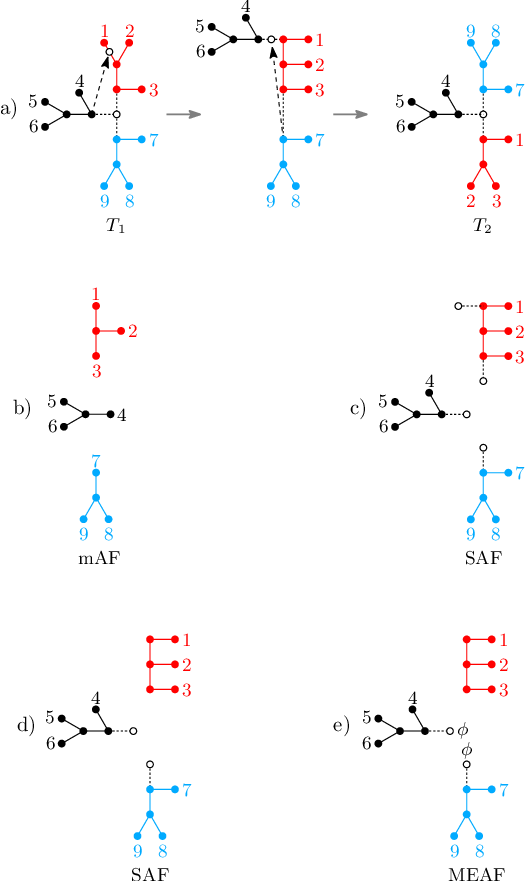}

	\caption{The intermediate structures used to find MEAFs.
	(a) An optimal replug path between a pair of trees, $T_1$ and $T_2$.
	(b) The single mAF of $T_1$ and $T_2$ can be obtained by cutting any two of the three dotted edges in $T_1$ or $T_2$. The $T_1$ and $T_2$ edge triples are both dead trees.
	(c) The SAF obtained by adding each of the removed nodes as sockets to the mAF.
	(d) An SAF with a set of candidate $\phi$-nodes is a subset of the SAF from c). Only sockets that correspond to nodes of both trees can remain fixed during a replug move, constraining our choices of dead tree edges.
	(e) The mEAF obtained by adding $\phi$ nodes is also the MEAF in this case.}
\label{fig:dead-tree-eaf}
\end{figure}

\thmenumeratemeafs*
\begin{proof}
	We apply $\alg{T_1, T_2, k}$ to enumerate the set of AFs of $T_1$ and $T_2$ with $k+1$ or fewer components.
	This requires $\OhOf{4^k n}$-time, by Theorem~\ref{thm:tbr}.
	For each of the at most $4^k$ AFs $F$, we apply $\replug{T_1, T_2, F, k}$ to enumerate the mEAFs, taking $\OhOf{4k \cdot 12^k n} = \OhOf{48^k n}$ time.
\end{proof}

\thmenumeratemeafsd*
\begin{proof}
Using a standard analysis of the recurrence relation $T(d) = 2T(d-1) T(d-2)$ for the number of recursive calls, and the fact that each call requires linear time, barring the cost of recursion, this algorithm takes $\OhOf{2.42^d n}$-time to determine whether a given edge set $E_1$ is compatible with an EAF of weight at most $d$.
Each branch of the search will terminate after $k$ recursive calls, so this will never take longer than the original version.
\end{proof}

\lemdeadtreeassignment*
\begin{proof}
We first prove that every such assignment $A$ is an EAF.
	Let $A$ be an assignment of $\phi$-nodes to sockets of $S$ that satisfies every dead tree in $T_1$ and $T_2$ (e.g. Fig~\ref{fig:dead-tree-eaf}(e)).
	Let $D$ be a dead tree of $T_1$ adjacent to a set $s$ of $n_D$ sockets of $S$ (Fig~\ref{fig:dead-tree-eaf}(a)).
The fact that $D$ is adjacent to $n_D$ sockets implies that $D$ was induced by removing $n_D - 1$ edges of $E_1$ from $T_1$.
Moreover, at least one socket in $s$ does not have a $\phi$ node.
Each of the $\phi$-nodes assigned to these sockets can be assigned to one of the $E_1$ edges within $D$.
An analogous assignment can be applied to the dead trees of $T_2$.
Applying this procedure iteratively to every dead tree of $T_1$ and $T_2$ results in an EAF of $T_1$ and $T_2$.

We now prove that every EAF induces such an assignment $A$.
	Let $F'$ be an EAF of $T_1$ and $T_2$ obtained by removing edge sets $E_1$ from $T_1$ and $E_2$ from $T_2$ (e.g. Fig~\ref{fig:dead-tree-eaf}(b)).
	Let $S$ be the SAF induced by $F'$ (e.g. Fig~\ref{fig:dead-tree-eaf}(c)).
	Let $A$ be the natural assignment of $\phi$-nodes to sockets of $S$ (Fig.~\ref{fig:dead-tree-eaf}(e)).
	$A$ is obtained by adding $\phi$-nodes to some subset of $S$ (e.g. Fig~\ref{fig:dead-tree-eaf}(d)).
Now, suppose that $A$ assigns a $\phi$ node to each of the $n_D$ sockets adjacent to some dead component of $T_1$.
Then both sides of one of the $n_D - 1$ edges of $T_1$ in $D$ and $E_1$ must have received a $\phi$ node, a contradiction.
\end{proof}

\lemreplugcnfletwo*
\begin{proof}
By Lemma~\ref{lem:dead-tree-assignment}, it suffices to find an assignment of $\phi$ nodes to $S$ that satisfy each dead component of $T_1$ and $T_2$.
Now, consider the set of constraints induced by these dead components.
We say that such a dead tree is \emph{saturated} if every endpoint of the dead tree is a candidate $\phi$-node.
We only consider constraints from saturated dead trees, as the constraints from unsaturated dead trees are trivially satisfied by excluding a socket that is not a $\phi$-node candidate.
Each dead tree implies a constraint of the form $(s_1 \cup s_2 \cup \ldots \cup s_{n_D})$, where the variable $s_i$ implies that socket $s_i$ is not assigned a $\phi$ node.
Moreover, a socket may be adjacent to at most two dead trees (one per tree), so each variable appears in at most 2 constraints.
There are at most $2k$ such constraints, with at most $2k$ variables.
Finally, each variable in a constraint is positive, that is no constraint includes $\lnot s_i$.
Thus, the full set of constraints is a boolean monotone $\cnfletwo$ formula.
\end{proof}

\lemsolveconstraints*
\begin{proof}
In the clause graph, each constraint of the formula is represented by a vertex.
Two vertices are connected by an edge if they share a variable.
A satisfying $\phi$-node assignment must include one true variable for each clause, and we wish to determine the minimum number of variables that must be true.
In other words, we wish to determine the minimum number of $\phi$-node candidates that do \emph{not} receive a $\phi$-node.
To do so, we can select a set of edges $C$ such that every vertex in the graph is adjacent to an edge in $C$: in other words, an edge cover.
Moreover, we wish to determine the minimum number of edges in any such edge cover---the edge cover problem.

The edge cover problem can be solved by finding a maximum matching (set of nonadjacent edges) and greedily adding additional edges~\cite{garey2002computers}.
This requires $\OhOf{\sqrt{V} E}$-time \cite{micali1980v}, for a graph with $V$ vertices and $E$ edges.
Our clause graph has at most $2k$ edges and vertices, and so requires $\OhOf{k^{1.5}}$-time to solve.
Observe that we can determine the set of $\phi$ nodes from the edge cover and thus also construct the~EAF.
\end{proof}

\thmcomputereplug*
\begin{proof}
We iteratively increase $k'=0, 1, \ldots, k$ until we find an MEAF of $T_1$ and $T_2$.
We apply $\alg{T_1, T_2, k}$ to enumerate the set of mAFs of $T_1$ and $T_2$ that can be obtained by cutting $k$ or fewer edges.
For each such mAF $F$ obtained by cutting $k'$ edges, we enumerate each of the at most $4^{k'}$ SAFs induced by $F$.
We apply the procedure in Lemma~\ref{lem:solve-constraints} to each SAF $S$ to determine whether an assignment of $\phi$-nodes to $S$ can be made that results in EAF of $T_1$ and $T_2$ with weight $\le k$.

The correctness of this procedure follows from Lemmas~\ref{lem:tbr},~\ref{lem:dead-tree-assignment}, and~\ref{lem:solve-constraints}.

The running time of the procedure follows from Theorem~\ref{thm:tbr} and Lemma~\ref{lem:solve-constraints}.
\end{proof}

\cortimereplug*
\begin{proof}
	We first prove the running time bound.
	Allen and Steel~\cite{allen01} proved that interleaving the subtree and chain reduction rules results in a pair of trees $T_1'$ and $T_2'$ with at most $28\dtbr{T_1, T_2}$ leaves.
	By Theorem~\ref{thm:replug-bound}, this is at most $28\dreplug{T_1, T_2}$ leaves.
	Substituting the size of the trees into Theorem~\ref{thm:compute-replug} results in the claimed running time bound.

	To prove that this method is correct, we note that the chain reduction proof of~\cite{whidden2016chain} holds whether one considers replug or SPR moves because it relies on finding an alternative set of moves to avoid breaking the common chain.
	Similarly, the subtree reduction also preserves the replug distance because optimal replug paths do not break common edges.
\end{proof}

\thmsprdistance*
\begin{proof}
This algorithm is guaranteed to find $T_2$, as no tree will ever be inserted into the priority queue twice with a priority determined by the same estimator function.
Moreover, each distance estimate is a lower bound on the true SPR distance, by Theorem~\ref{thm:replug-bound}.
Therefore, as with a standard A* search, the correct distance is guaranteed to be returned.
Finally, this method will never compute a replug or TBR distance larger than $\dspr{T_1, T_2} + 1$, as we will never expand a tree $T \ne T_1$ with an estimator greater than $\dspr{T_1, T_2}$.
This is an important consideration, as the running time of our TBR and replug distance calculations are exponential with respect to the distance computed.

The algorithm explores $Y$ trees, each of which has replug distance at most one greater than $\dspr{T_1, T_2}$.
Therefore the time required to compute estimators for each of these trees is $\OhOf{16^{d+1}((d+1)^{1.5} + n)}$ by Corollary~\ref{cor:time:replug}.
This is $\OhOf{16^d d^{1.5}}$ for reduced trees with $n = \OhOf{d}$.
All other operations on each tree can be carried out in time that is asymptotically smaller than the cost of a replug distance computation.
Therefore the running time is bounded as claimed.
\end{proof}

\begin{figure*}
	\hspace*{\stretch{1}}
	\includegraphics[width=\textwidth]{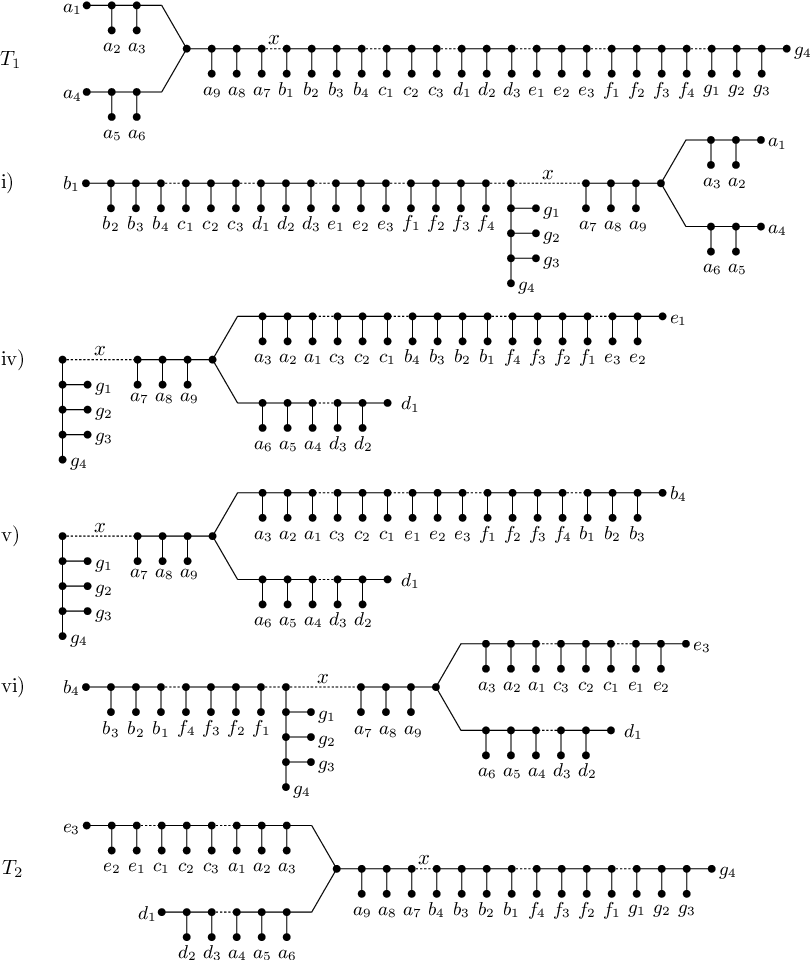}
	\hspace*{\stretch{1}}

	\caption{A portion of an SPR path between two trees $T_1$ and $T_2$ for which every optimal SPR path underlain by the sole MAF moves an endpoint of the same edge twice.
		The trees have SPR distance 7 and only one MAF with 6 components.
		The MAF can be obtained by removing the dotted edges.
		The SPR path shown, for example, moves the endpoint of edge $x$ closest to leaf $a_7$ of the $a$ component twice, in the first and last move.
		Note that 3 moves are applied to move from i) to iv).
	}
	\label{fig:move-twice-counterexample}
\end{figure*}

\begin{figure*}
	\hspace*{\stretch{1}}
	\includegraphics[width=0.85\textwidth]{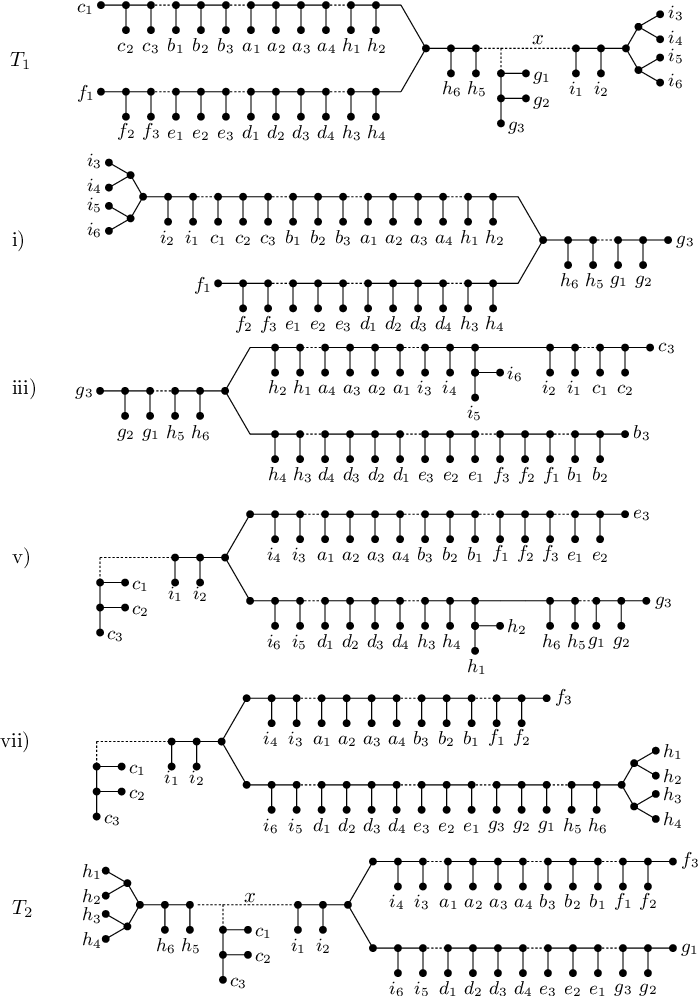}
	\hspace*{\stretch{1}}

	\caption{Two trees $T_1$ and $T_2$ for which every optimal SPR path breaks a common path.
		The trees have SPR distance 8 and every optimal SPR path modifies only the dotted edges, corresponding to the sole MAF.
		The SPR path shown, for example, breaks the common path $x$ between the $h$ and $i$ components in the first move and then reforms this path in the eighth move.
		Note that each tree other than the second and last is a result of applying two SPR moves.
	}
	\label{fig:break-common}
\end{figure*}

\end{document}